\documentclass[letterpaper,journal]{IEEEtran}
\usepackage{amsmath,amsfonts}
\usepackage{algorithmic}
\usepackage{algorithm}
\usepackage{array}
\usepackage[caption=false,font=normalsize,labelfont=sf,textfont=sf]{subfig}
\usepackage{textcomp}
\usepackage{stfloats}
\usepackage{url}
\usepackage{verbatim}
\usepackage{graphicx}
\usepackage{cite}

\usepackage{graphics} 
\usepackage{epsfig} 
\usepackage{mathptmx} 
\usepackage{amsthm}
\usepackage{amsmath} 
\usepackage{amssymb}  
\usepackage{cuted}
\usepackage{float}
\usepackage{color}
\usepackage{url}
\usepackage{bm}

\usepackage{pifont}
\usepackage{makecell}
\usepackage{booktabs}
\usepackage{tikz}
\usetikzlibrary{shapes, arrows.meta, positioning, decorations.pathreplacing, patterns}
\usepackage{tabularx}
\usepackage[export]{adjustbox}   
\usepackage{cite}

\newtheorem{theorem}{Theorem}

\newtheorem{lemma}{Lemma}

\newtheorem{remark}{Remark}

\begin{document}

\title{STARDIS: Strategic Scheduling and Deceptive Signaling for Satellite Intrusion Detection System Deployment}

\author{Yuzhou Xiao\IEEEauthorrefmark{1}\thanks{\IEEEauthorrefmark{1}Linan Huang and Yuzhou Xiao are co-first authors and contributed equally to this work.},
	Linan Huang\IEEEauthorrefmark{1}\textsuperscript{,}\IEEEauthorrefmark{2},~\IEEEmembership{Member,~IEEE},
	Jiachen Sun,~\IEEEmembership{Student Member,~IEEE}, 
	Peilong Liu,~\IEEEmembership{Member,~IEEE}, 
	Chunxiao Jiang,~\IEEEmembership{Fellow,~IEEE}, 
	Linling Kuang\IEEEauthorrefmark{2}\thanks{\IEEEauthorrefmark{2}Corresponding authors: Linan Huang (email: huanglinan@tsinghua.edu.cn) and Linling Kuang (email: kll@tsinghua.edu.cn).},~\IEEEmembership{Member,~IEEE} 

    \textit{This article has been accepted for publication in the IEEE Journal on Selected Areas in Communications. Copyright may be transferred without notice, after which this version may no longer be accessible.}
    
	\thanks{This work was supported in part by the National Natural Science Foundation of China under Grant 62401313, Grant 62325108, and Grant 62341131.}
	\thanks{The Project Supported by National Natural Science Foundation of China No.62341109, No.62341106, Shanghai Municipal Science and Technology Major Project, and Tsinghua University Initiative Scientific Research Program.}

	\thanks{Y. Xiao is  with the School of Automation, Beijing Institute of
Technology, Beijing 100081, China (e-mail: y.xiao@bit.edu.cn)} 	
        \thanks{
            L. Huang, P. Liu, C. Jiang, and L. Kuang are with the Beijing National Research Center for Information Science and Technology (BNRist), Tsinghua University, Beijing 100084, China (e-mail:huanglinan@tsinghua.edu.cn; plliu@tsinghua.edu.cn; jchx@tsinghua.edu.cn; kll@tsinghua.edu.cn\}). 
		}
		\thanks{J. Sun is with the Department of Electronic Engineering, Tsinghua University, Beijing 100084, China (e-mail: sjc20@mails.tsinghua.edu.cn).}
		\thanks{C. Jiang and L. Kuang are also with State Key Laboratory of Space Network and Communications, Tsinghua University, Beijing, China.}
	}

\markboth{IEEE Journal on Selected Areas in Communications,~Vol.~xx, No.~x, July~2025}%
{Shell \MakeLowercase{\textit{et al.}}: A Sample Article Using IEEEtran.cls for IEEE Journals}


\maketitle

\begin{abstract}
	Satellite communication networks operate under stringent computational constraints and are susceptible to sophisticated cyberattacks. This paper introduces a novel defense framework that decouples security optimization into ground-based analysis and onboard real-time execution. In the long-term loop, the ground segment processes historical data to estimate key statistical parameters of the task environment. 
	\textbf{{Additionally, we incorporate the time-varying characteristics of satellite wireless links to account for the dynamic communication context.}} 
	In the short-term loop, the satellite employs a receding horizon optimization that models dynamic task arrivals and maximizes a utility function considering detection rates and resource costs. To counter intelligent adversaries interception, we introduce a deception mechanism using Bayesian persuasion theory. By strategically manipulating the short-term action sequences in the telemetry downlink, we mislead an external attacker's beliefs. We mathematically model the attacker's optimal response under \textbf{{channel uncertainty}} and demonstrate that our framework significantly reduces attacker utility. The approach's effectiveness is formally proven using Lyapunov theory.
\end{abstract}


\section{Introduction}

Satellite internet is rapidly evolving into a cornerstone of global connectivity, enabling applications ranging from 6G and Internet of Things (IoT) to mission-critical services across industries \cite{Choi2024, Torrens2024}. The transition from isolated, closed-loop architectures to large-scale, IP-based infrastructures has critically broadened the attack surface of satellite systems \cite{Manulis2021, Liu2016, Hao2023, Li2020}. High-profile incidents, such as the malware attack on Viasat modems, {highlight} the {real} cyber-physical threats to modern satellite networks \cite{Liu2016, Yu2024, Swope2025, Javadpour2024}.

Securing satellite internet demands a paradigm shift from traditional, terrestrial-centered security models \cite{Rodrigues2023}. While ground segments and user segments are essential components, they represent inherently flawed vantage points for comprehensive threat detection. Ground-based Intrusion Detection Systems (IDS) are fundamentally constrained by the high-latency and bandwidth-limited satellite-to-ground link, rendering them ineffective against real-time, intra-satellite threats. Concurrently, user terminals, being geographically dispersed and often physically insecure, are susceptible to compromise and thus cannot serve as a reliable foundation for system-wide security monitoring.

This challenge is magnified by the accelerating trend of migrating core functionalities directly into the space segment. For instance, next-generation satellite architectures are increasingly designed to host core network functions and advanced data processing on board \cite{5GAmericas2025NTN}. Such a shift means that sophisticated attacks can be initiated and executed entirely within the space segment, bypassing terrestrial defenses completely. Therefore, to ensure timely detection and enable autonomous response, embedding security capabilities directly within the satellite payload is not merely an enhancement, but a critical necessity \cite{Mao2024}. This directly motivates the development of on-board IDS. The architecture of the user segment, space segment, ground segment, and the reasons for deploying IDs on board are shown in Fig. \ref{fig:multidomainStructure}.

\begin{figure}[htbp]
	\centering
	\includegraphics[width=1 \columnwidth]{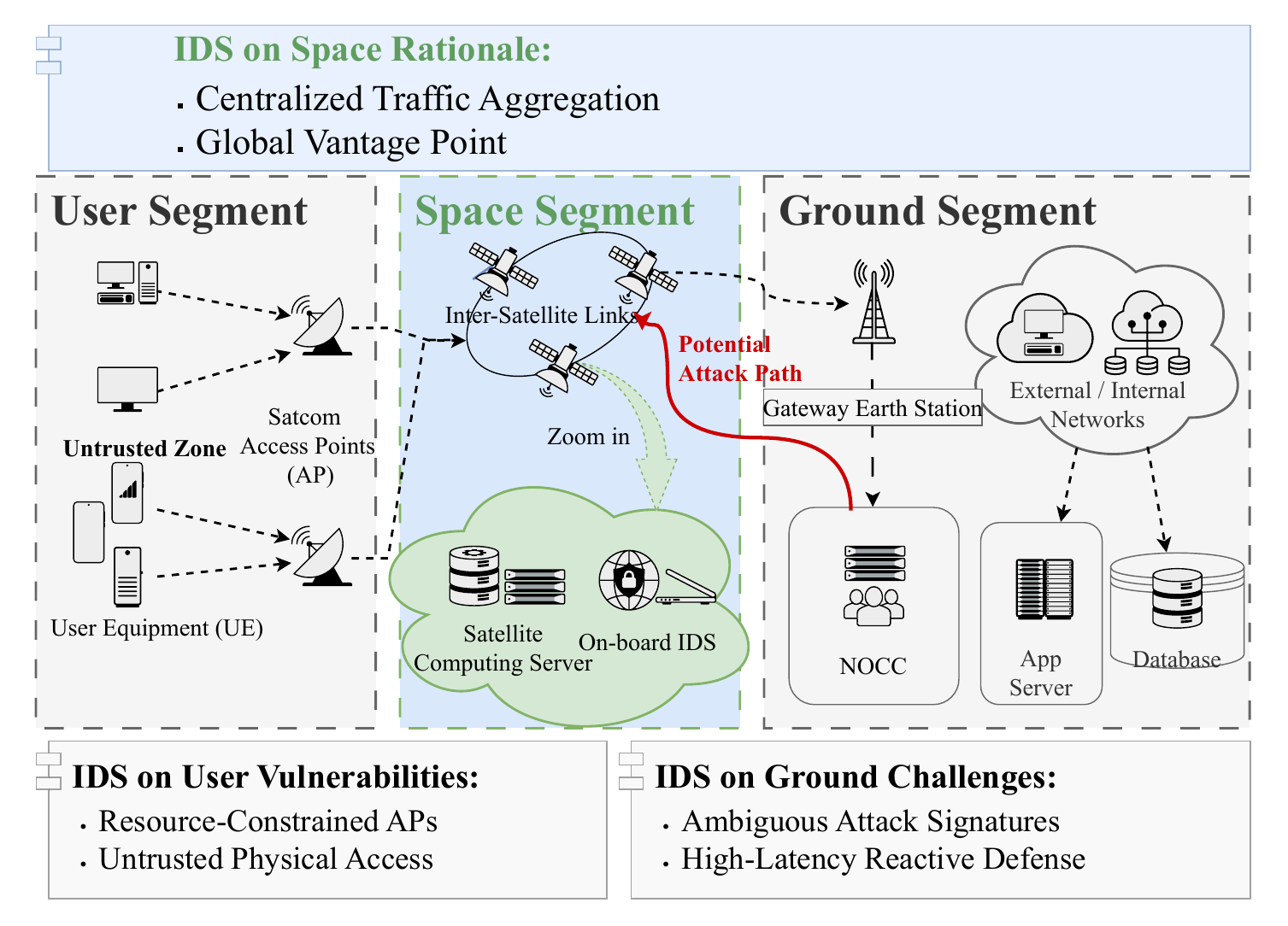}
	\caption{ 
		The necessity of deploying IDS in the space segment.
}
\label{fig:multidomainStructure}

\end{figure}
By placing an IDS on board, threat data can be processed in situ, allowing for immediate defensive actions while conserving bandwidth \cite{Zheng2018, Yao2025}. However, this approach creates a fundamental conflict: the high computational demands of modern IDS are at odds with the strict Size, Weight, and Power (SWaP) constraints of satellite hardware \cite{Zheng2018, Roy2024}. 
To address the strict Size, Weight, and Power (SWaP) constraints of satellite platforms, a two-pronged approach is envisioned. A foundational strategy involves adapting lightweight IDS techniques, originally conceived for resource-limited settings like terrestrial IoT, to engineer efficient on-board detection modules \cite{Roy2024, Sai2021}. However, a more profound strategy emerges from exploiting the intrinsic variability of satellite mission workloads. This allows security to be reframed as a schedulable, mission-aware task within a unified resource management framework \cite{Sprunt1989, Li2023a}. Consequently, the IDS can dynamically modulate its operational intensity, escalating to high-fidelity, computationally-intensive analysis during periods of resource abundance, and reverting to baseline, low-overhead monitoring when mission-critical functions demand priority. This strategic scheduling fosters a truly resource-aware and service-aware security posture, dynamically aligning defensive rigor with both operational tempo and real-time resource capacity.

Yet this optimization itself introduces a critical vulnerability: by design, an optimal policy tends to be predictable. A strategic adversary can exploit this predictability, learning the scheduling patterns and launching attacks during low-surveillance windows—effectively transforming resource-efficient defense into a time-sensitive attack surface.
This exposes a fundamental asymmetry: while the defender seeks efficiency, the attacker exploits regularity. To counter such strategic exploitation, a purely reactive or static defense is insufficient. Instead, deceptive defense emerges as a compelling paradigm—one that breaks the attacker’s assumptions by introducing controlled unpredictability into the system. By strategically obfuscating the IDS scheduling policy, or even simulating false detection patterns, the satellite can manipulate the adversary’s belief and induce suboptimal attack timing. 

{Unlike wired terrestrial networks, satellite communications rely on open wireless channels that are inherently vulnerable to interception. However, these channels are also subject to rapid quality fluctuations due to high mobility, path loss, and atmospheric fading. An intelligent attacker attempting to intercept telemetry to plan attacks faces significant physical layer challenges. The received signal may be delayed, distorted, or entirely missing due to deep fading or shadowing effects. Previous works often ignored these physical constraints, assuming perfect communication links. This paper fills that gap by analyzing how channel fluctuations impact the attacker's information acquisition and proposing a method to actively use these fluctuations—alongside artificial delays—as a defensive tool.}

To this end, we propose STARDIS (Strategic Scheduling and Deceptive Signaling), a novel game-theoretic framework integrating two synergistic components. 
The first component, \textbf{STAR (Strategic Scheduling)}, establishes a unified resource allocation framework to resolve the inherent contention between routine mission operations and high-priority security tasks. This framework formally models the satellite's workload by co-scheduling two distinct task types—predictable mission workflows and stochastic security responses—as competitors within a single optimization problem \cite{Li2023a}. The solution yields an optimal, resource-aware IDS policy that intelligently governs the allocation of the satellite's finite computational resources.
Building upon this, the second component, \textbf{DIS (Deceptive Signaling)}, introduces a protective layer of strategic deception. This layer is founded on Bayesian Persuasion theory \cite{Kamenica2011, Bergemann2016, Sayin2021}, where the defender (the ``Sender'') commits to a carefully engineered signaling scheme embedded in public telemetry. While correlated with the true IDS state, the signals are designed to manipulate a rational attacker's (the ``Receiver's'') beliefs \cite{Crawford1982, Xu2015} and persuade them to choose an attack strategy favorable to the defender \cite{Zhu2013}. As a key theoretical contribution, we leverage Lyapunov stability theory to rigorously prove the stability of the resulting persuasive equilibrium, guaranteeing the long-term suppression of attacker utility and enhancement of the defender's security posture.
Fig. \ref{fig:framework_overview} illustrates our proactive defense framework.
\begin{figure}[t]
	\centering
	\includegraphics[width=1 \columnwidth]{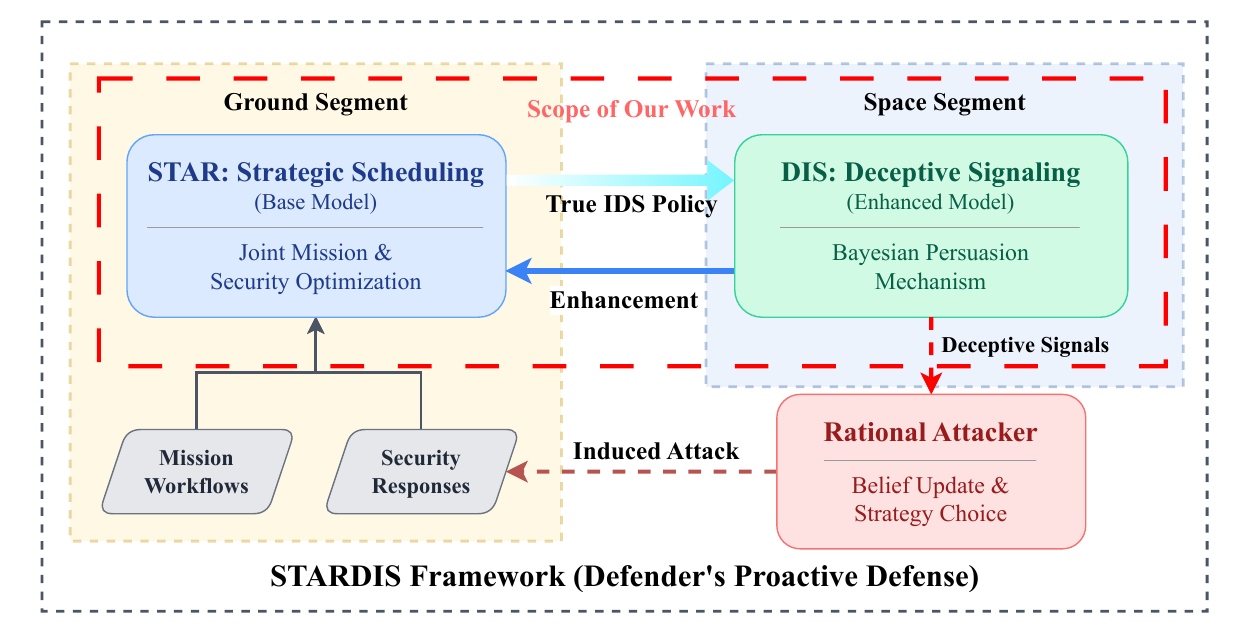}
	\caption{
	Overview of the STARDIS proactive defense framework. The ground-based STAR module optimally co-schedules mission and security tasks to determine a true security policy. This policy informs the space-based DIS module, which uses Bayesian persuasion to generate deceptive public signals. These signals manipulate a rational attacker's beliefs, steering their actions toward outcomes that are advantageous for the defender and feeding back into the defense cycle.}
	\label{fig:framework_overview}
	
\end{figure}

The main contributions of this paper are summarized as follows:
\begin{enumerate}
	\item \textbf{Unified Security and Mission-Task Optimization Framework:} We propose a novel ground-satellite coordinated defense framework that integrates intrusion detection into the satellite's native multi-task scheduling system. By abstracting the IDS as schedulable periodic and aperiodic tasks, our approach enables dynamic optimization of detection performance versus resource consumption. The long-term optimal policy is computed on the ground and executed via a lightweight on-board decision-maker, realizing a resource-aware and responsive security architecture.
	
	\item \textbf{Strategic Deception via Channel-Aware Bayesian Persuasion:} We introduce a strategic deception mechanism based on Bayesian persuasion theory that integrates a realistic wireless channel model accounting for fading, noise, and latency. This approach models the dynamic interplay of beliefs between the defender and the attacker under channel uncertainty. It allows the satellite to generate a persuasive signaling policy that manipulates the attacker's beliefs, guiding them toward suboptimal actions by turning the physical properties of the link into a defensive tool.
	
	\item \textbf{Credibility-Constrained Deception and Stability Guarantee:} To resolve the inherent conflict between deception and trustworthiness, we design a novel credibility budget using Kullback-Leibler (KL) divergence to constrain the information manipulation. Furthermore, we provide a rigorous theoretical guarantee for the long-term stability and effectiveness of our strategy. By constructing a Lyapunov function from the attacker's maximum expected utility, we formally prove that our signaling policy forces the system to converge to a state where the adversary's impact is minimized.
\end{enumerate}

\section{Related Work}
Our research builds on two primary domains: the challenge of implementing robust security within the resource-constrained environment of satellites, and the use of advanced game-theoretic methods for proactive defense. {Specifically, the rapid evolution of Space-Air-Ground Integrated Networks (SAGIN) facilitates ubiquitous connectivity but simultaneously introduces complex security challenges due to the integration of heterogeneous segments. In such dynamic environments, efficient collaborative offloading and resource allocation are fundamental not only for mission performance but also for sustaining robust defense mechanisms \cite{Huang2025}.} This section reviews the state of the art in these areas to contextualize our contributions.

\subsection{Resource-Aware Intrusion Detection in Satellite Networks}
The shift to IP-based Low Earth Orbit (LEO) mega-constellations has made satellite networks critical infrastructure, yet has also exposed them to terrestrial cyber threats and expanded the attack surface \cite{Manulis2021}. Documented incidents like the Viasat modem attack, alongside persistent jamming and spoofing, confirm that sophisticated actors are actively targeting these systems with threats such as Denial-of-Service (DoS) and selective forwarding attacks \cite{Liu2016, Swope2025, Yu2024, Javadpour2024, Sun2025, Sedjelmaci2024}. While some research addresses these risks via physical layer security or cryptography \cite{Liu2016}, our focus is on real-time, on-board operational security.

Deploying an IDS on a satellite is vital for timely response but is challenged by severe SWaP constraints and the use of less powerful, radiation-hardened processors \cite{Zheng2018}. This has spurred research into ``lightweight'' IDS, which minimize computational footprints using techniques like feature selection, often borrowing from IoT and vehicular network domains \cite{Roy2024, Sai2021, Sedjelmaci2024}. However, these designs typically assume a static resource budget and are not adaptive. Concurrently, the field of dynamic resource management in satellites has well-established models for scheduling periodic and aperiodic tasks to optimize resource utilization \cite{Hao2023, Li2020, Sprunt1989, Li2023a}. The primary gap in the literature, which our work addresses, is that these resource management frameworks do not treat security as a schedulable, resource-consuming task. We bridge this gap by designing an IDS whose operational intensity is the output of a dynamic optimization problem, thus making security truly resource-aware.

\subsection{Game-Theoretic Deception for Proactive Defense}
Against strategic adversaries who learn and adapt, reactive security is insufficient. Game theory provides a mathematical framework for modeling these conflicts, leading to proactive defenses like Moving Target Defense and cyber deception \cite{Tambe2011, Carroll2011, Javadpour2024, Ma2022, Huang2020}. However, a sophisticated attacker may anticipate simple deception, rendering it ineffective \cite{Huang2021}. This {requires} more advanced methods that exploit the inherent information asymmetry in cyber warfare \cite{Xu2015}.

Our work advances this paradigm by employing Bayesian Persuasion (BP), a form of information design \cite{Kamenica2011}. In the BP framework, a defender (``Sender'') commits to a signaling policy that is correlated with the true state of the world, designed to persuade a rational, Bayesian attacker (``Receiver'') to take actions favorable to the defender \cite{Bergemann2016}. This commitment distinguishes BP from classic cheap talk signaling games \cite{Crawford1982}. While the concept has been explored for cyber-physical systems \cite{Sayin2021, Zhu2013}, its application to protect an adaptive defense policy from being learned and exploited is novel. {We extend this interaction into a Bayesian persuasion framework} and, critically, provide rigorous theoretical guarantees for its stability using Lyapunov theory, a tool for analyzing equilibria in dynamic systems \cite{Neely2010, Clempner2018}. 
This synthesis of resource optimization, information-theoretic deception, and stability analysis represents a holistic advance in satellite security. Table \ref{tab:related_work_comparison} contrasts our proposed framework with existing works across several key capability dimensions.

\begin{table*}[t]
	\centering
	\caption{{Comparison of Our Work with the State of the Art in Communication Security}}
	\label{tab:related_work_comparison}
	\begin{tabularx}{\textwidth}{l >{\centering\arraybackslash}X >{\centering\arraybackslash}X >{\centering\arraybackslash}X >{\centering\arraybackslash}X >{\centering\arraybackslash}X}
		\toprule
		Existing Works & Satellite-Specific Context & Dynamic Resource-Awareness & Models Strategic Adversary & Info-Theoretic Deception & Theoretical Guarantee \\
		\midrule
		\cite{Javadpour2024} & \ding{55} & \ding{55} & \ding{51} & \ding{51} & \ding{55} \\
		\addlinespace
		\cite{Roy2024} & \ding{55} & \ding{55} & \ding{55} & \ding{55} & \ding{55} \\
		\addlinespace
		\cite{Sai2021} & \ding{55} & \ding{55} & \ding{55} & \ding{55} & \ding{55} \\
		\addlinespace
		\cite{Li2023a} & \ding{51} & \ding{51} & \ding{55} & \ding{55} & \ding{51} \\
		\addlinespace
		\cite{Sayin2021} & \ding{55} & \ding{55} & \ding{51} & \ding{51} & \ding{51} \\
		\addlinespace
		\cite{Huang2021} & \ding{55} & \ding{55} & \ding{51} & Partial & \ding{51} \\
		\addlinespace	
		This work & \ding{51} & \ding{51} & \ding{51} & \ding{51} & \ding{51} \\
		\bottomrule
	\end{tabularx}
\end{table*}

\section{System Modeling and Architecture}

To realize an intelligent on-board defense, we should first accurately model the operational environment of the satellite, including its hardware architecture, task characteristics, communication links, and the threat model. This section establishes the foundational \textit{system model} that describes how the ground and space segments collaborate, how tasks are prioritized and queued under heterogeneous resource constraints, and how the physical channel and adversary are mathematically characterized.

\subsection{Two-Timescale Architecture}

Our strategy is built upon a two-timescale architecture, as depicted in Fig.~\ref{fig:DualLoop}, which decouples the scheduling problem into a ground-based, long-term analysis and an onboard, short-term real-time optimization. This separation is crucial for managing computational complexity while maintaining adaptability in a dynamic satellite environment.

\begin{figure}[hpbt!]
	\centering
	\includegraphics[width=1 \columnwidth]{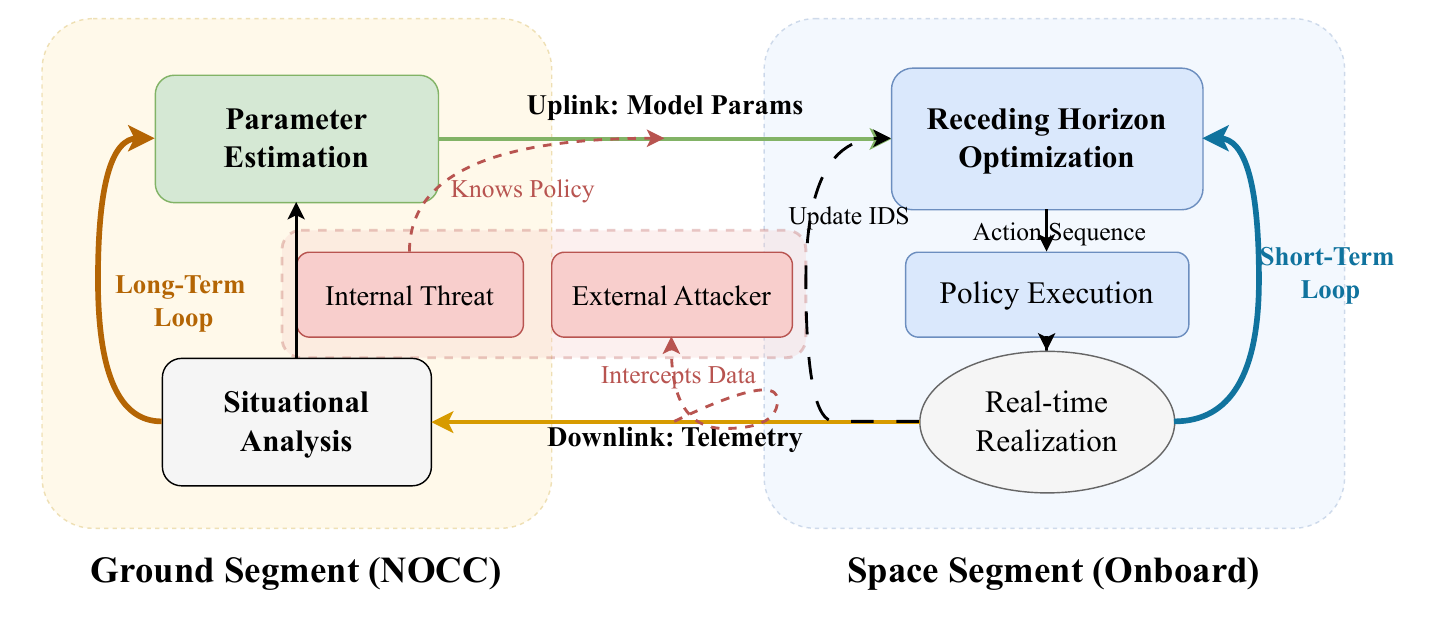}
	\caption{The proposed two-timescale architecture. The ground segment performs long-term parameter estimation and uplinks model parameters to the satellite. Onboard, the space segment employs short-term receding horizon optimization to generate real-time action sequences. The framework considers an internal threat accessing the policy and an external attacker intercepting the telemetry.}
	\label{fig:DualLoop}
\end{figure}

\noindent\textbf{Ground-Based Computation (Long-Term Loop):} {The ground segment, functioning as the Network Operations Control Center (NOCC), is responsible for the following offline}, long-timescale operations:

\begin{enumerate}
	\item \textbf{Data Analysis:} It processes extensive historical mission data to build and validate statistical models of the operational environment.
	
	\item \textbf{Parameter Estimation:} From this analysis, it estimates key statistical parameters, most notably the average arrival rates $\lambda_j$ for aperiodic tasks. This process establishes the long-term operational requirements needed to ensure system stability.
	
	\item \textbf{Parameter Uplink:} The output is not a fixed schedule, but a set of validated model parameters which are then uplinked to the satellite to guide its real-time decisions.
\end{enumerate}

{It is important to note that due to the significant difference in timescales, the ground segment does not dynamically react to instantaneous attacker actions within the short-term optimization window. Consequently, in our simulation model, the parameters provided by the ground are treated as fixed constants over the execution horizon $N_t$.}

\noindent\textbf{Onboard Real-Time Optimization (Short-Term Loop):} 
The satellite operates on a fast timescale, employing a receding horizon control strategy. 
At the beginning of each decision horizon, the onboard scheduler solves a deterministic optimization problem for the upcoming time window. 
The primary \textbf{inputs} to this problem are the statistical model parameters provided by the ground (e.g., average task arrival rates), fixed task specifications (e.g., resource demands and execution deadlines), and the real-time system state. 
The \textbf{outputs} of the optimization are the optimal binary scheduling sequences that dictate which mission and security tasks to execute in each time slot. 
This approach makes the problem computationally tractable for onboard processors while allowing the satellite to react swiftly to operational realities.

{
	Specifically, the onboard resource allocation problem involves binary decision variables and nonlinear coupling, falling into the class of Mixed-Integer Nonlinear Programming (MINLP), which is generally NP-hard.
	Furthermore, the derivation of the optimal signaling policy requires solving a global convex optimization problem over high-dimensional belief spaces, while the channel-aware defense strategy requires complex stochastic integration of Shadowed Rician fading statistics.
	Performing these heavy computations onboard is prohibitive due to the strict SWaP constraints and the limited capability of radiation-hardened satellite processors \cite{Tang2021, Xiu2025}.
	Therefore, offloading these computationally intensive tasks to the ground segment is necessary, allowing the satellite to focus solely on the lightweight receding horizon execution.}

\subsection{Unified Task and Queue Model}

{Unlike terrestrial data centers with abundant homogeneous servers, satellites operate with strictly constrained and heterogeneous hardware. We model the satellite's onboard computing platform as a system comprising a set of distinct resource types, denoted by $\mathcal{K} = \{1, \dots, K\}$ (e.g., General Purpose Processors (CPU) for logic control and Field Programmable Gate Arrays (FPGA) or GPUs for parallel signal processing).}
The system operates over a discrete-time horizon of $N_t$ slots, indexed by $t \in \{0, 1, \dots, N_t-1\}$.

Tasks are defined along three orthogonal dimensions:
 \textbf{task nature}, \textbf{priority level}, and \textbf{arrival pattern}. This multi-dimensional classification allows for a more granular and realistic representation of the satellite's operational state. The 3D task classification diagram can be seen in Fig. \ref{fig:3DTask}.
\begin{figure}[hpbt!]
	\centering
	\includegraphics[width=1 \columnwidth]{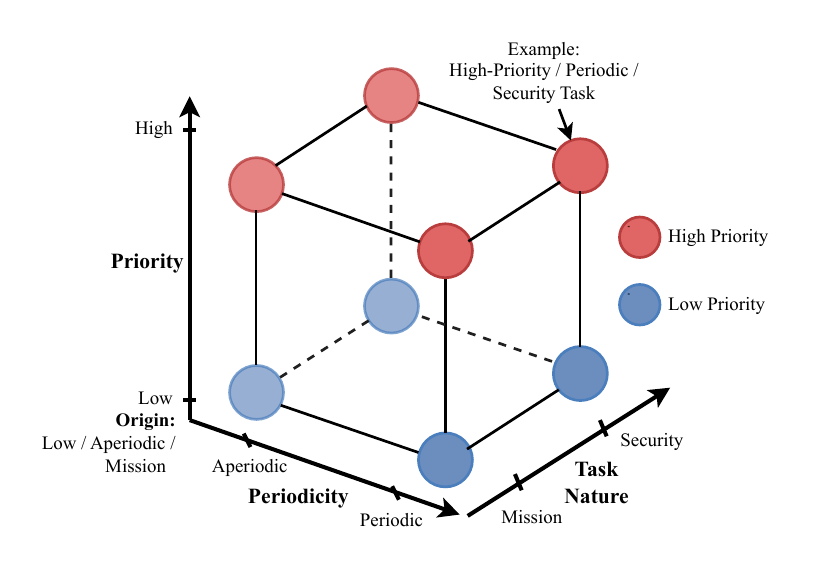}
	\caption{ Task classification for 3 dimensions. }
	\label{fig:3DTask}
	
\end{figure}

\begin{enumerate}
	\item \textbf{Task Nature:} This dimension distinguishes between the fundamental objectives of a task.
	\begin{itemize}
		\item \textit{Mission Tasks ($\mathcal{T}_M$):} These encompass the satellite's primary operational functions, such as remote sensing data processing, user-requested computations, or delay-tolerant data relaying.
		\item \textit{Security Tasks ($\mathcal{T}_S$):} These are initiated by the onboard IDS to protect the satellite. Examples include deep packet inspection, cryptographic verification, anomaly analysis, or executing countermeasures.
	\end{itemize}
	
	\item \textbf{Priority Level:} This dimension dictates a task's scheduling precedence and its ability to preempt other tasks.
	\begin{itemize}
		\item \textit{High-Priority (Preemptive):} These tasks are critical and can interrupt the execution of lower-priority tasks to gain immediate access to computational resources.
		\item \textit{Low-Priority (Preemptible):} These tasks can be suspended or aborted by the scheduler to accommodate high-priority tasks.
	\end{itemize}
	
	\item \textbf{Arrival Pattern:} This dimension characterizes the predictability of a task's arrival.
	\begin{itemize}
		\item \textit{Periodic (Deterministic):} These tasks arrive at regular, predictable intervals. Their scheduling can be planned in advance.
		\item \textit{Aperiodic (Stochastic):} These tasks arrive unpredictably, often modeled by a stochastic process (e.g., a Poisson distribution), reflecting events like ad-hoc user requests or random cyberattacks.
	\end{itemize}
\end{enumerate}

The combination of these three dimensions creates a spectrum of eight distinct task categories. This framework moves beyond a rigid mission-versus-security view, enabling more sophisticated scheduling decisions. For instance, a critical but non-security-related mission task (e.g., an emergency orientation adjustment) can be modeled as a \textit{high-priority, aperiodic mission task}. Conversely, a routine security log audit could be a \textit{low-priority, periodic security task}. 

Consequently, the most pressing threats are encapsulated as \textit{high-priority, aperiodic security tasks}, which demand immediate, preemptive execution to preserve system integrity. {Specifically, a primary example of such urgency is the real-time validation of telecommands. As emphasized by He \textit{et al.} \cite{He2025}, the injection of spoofed uplink commands can result in unauthorized orbit maneuvers or payload misuse, potentially leading to the irreversible physical loss of satellite control; thus, command verification should be strictly enforced as a hard real-time constraint before execution. In parallel, the onboard system should contend with software-driven threats such as malware propagation. As noted in comprehensive surveys by Guo \textit{et al.} \cite{Guo2022}, modern satellites within space-air-ground integrated networks are increasingly vulnerable to ransomware-like behaviors. Upon the detection of such anomalies, the immediate isolation of the compromised process is mandatory to contain the breach and prevent system-wide paralysis. These tasks are therefore assigned the highest scheduling priority to mitigate risks that could otherwise be catastrophic.}

Each individual task $S_j$ across all categories is uniformly characterized by a tuple:
\begin{equation}
	S_j := ({\text{req}_j,\;a_j,\;d_j,\;p_j}),
\end{equation}
where ${\text{req}_j}$ is the request arrival time, ${a_j}$ is the earliest feasible start time, ${d_j}$ is its execution deadline, and ${p_j}$ is the required processing time {(i.e., the net resource occupancy duration). In Fig. \ref{fig:system_model}, the ``Task Execution'' bar visually represents the total turnaround time, which may be longer than $p_j$ if the task experiences suspension due to preemption.}
A task is defined to be schedulable if ${d_j - a_j \ge p_j}$.

As depicted in Fig. \ref{fig:system_model}, all incoming tasks first enter a global input queue $Q$. An admission controller then assesses their schedulability (i.e., if $d_i - t \ge p_i$ at the current time $t$) before moving them to buffers.
{To reflect specific hardware affinities, the scheduler dispatches tasks to the appropriate resource queues (e.g., dispatching image matrix operations to the FPGA queue).}
\begin{figure}[hpbt!]
	\centering
	\includegraphics[width=1 \columnwidth]{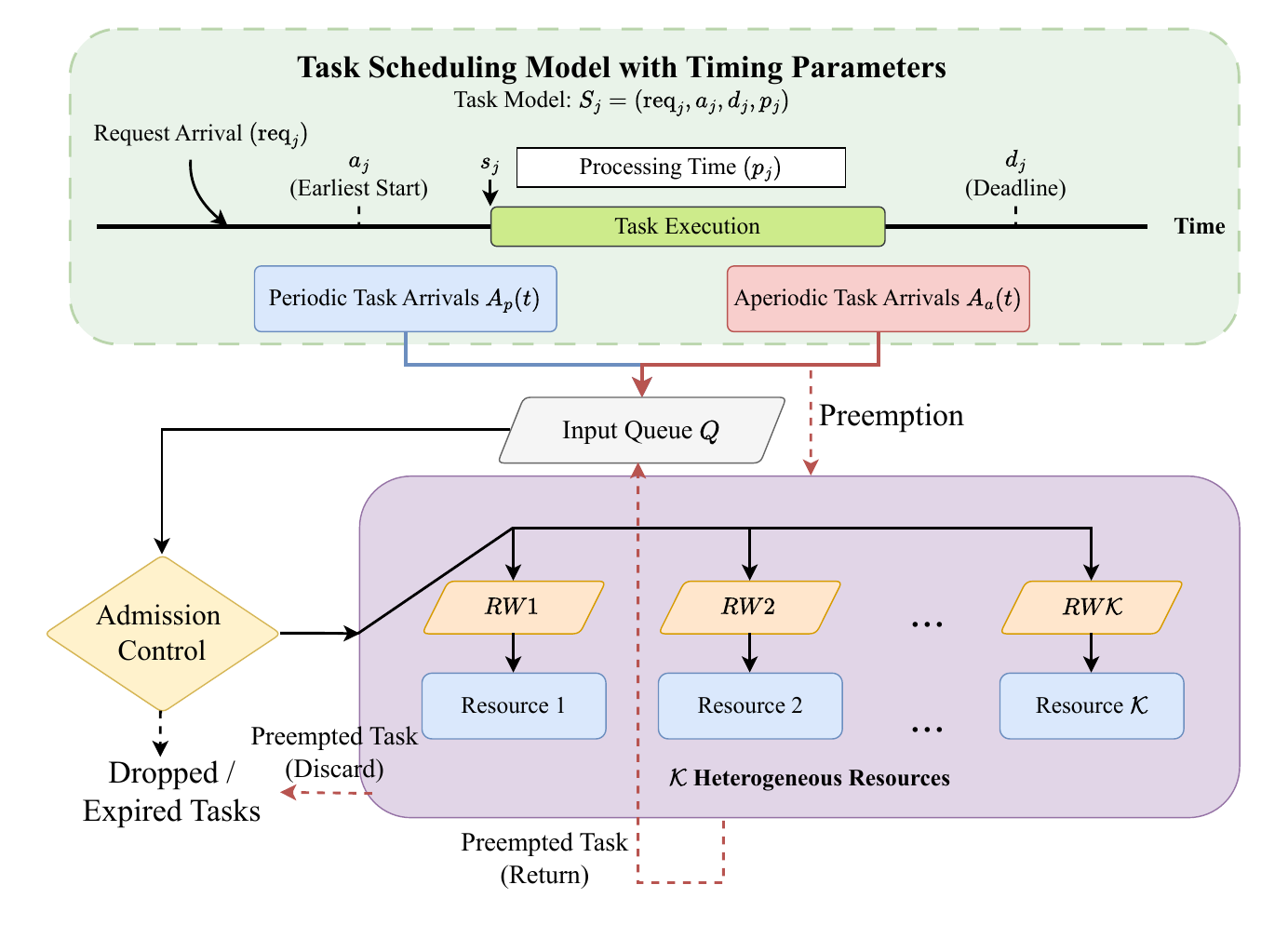}
	\caption{
		The proposed unified scheduling model. Periodic and aperiodic tasks arrive in a common queue \(Q\). Feasible tasks are admitted to buffers, {designated as Ready Workload (\(RW\)), for heterogeneous resources.} High-priority tasks can preempt ongoing low-priority tasks, ensuring immediate response to critical events. The timing parameters for any task \(S_j\) are shown below.}
	\label{fig:system_model}
\end{figure}
A critical feature of our unified model is \textbf{preemptive priority}.
The arrival of any high-priority task grants it immediate access to a computational resource, even if it means preempting an ongoing, lower-priority task.
The preempted task can be either discarded or returned to the buffer for later rescheduling.
This mechanism ensures that urgent events, whether mission-critical or security-related, are handled with minimal delay, directly modeling the core challenge of balancing operational integrity with security assurance.
{While Fig. \ref{fig:system_model} depicts preemption initiated by an aperiodic task to highlight the handling of unpredictable bursts, the preemption logic is general and applies to any high-priority task regardless of its arrival pattern.}

\subsubsection{Resource Demand and System State Definition}
{To capture hardware heterogeneity, each task $j$ is characterized by a resource demand vector $\mathbf{r}_j = \{r_{j,k}\}_{k \in \mathcal{K}}$, where $r_{j,k} \in [0,1]$ is the normalized resource usage.
The scheduler's decisions are represented by binary variables: $x_{j}(t) \in \{0,1\}$ for mission tasks and $x_{\text{scan}}(t) \in \{0,1\}$ for the IDS scan.

Crucially, the instantaneous load of the satellite determines the "opportunity" for attacks. We define the \textbf{idle capacity} at time $t$, denoted by $z(t)$, as the minimum fraction of available resources across all types:
\begin{equation}
	\label{eq:def_idle_capacity_sec3}
	z(t) = \min_{k \in \mathcal{K}} \left( 1 - \sum_{j \in \mathcal{J}(t)} r_{j,k}\,x_{j}(t) - r_{s,k}\,x_{\text{scan}}(t) \right).
\end{equation}
A low value of $z(t)$ implies the system is congested, creating a vulnerability window that the attacker seeks to exploit, while a high $z(t)$ indicates available capacity for security scanning.}

\subsection{Time-Varying Wireless Channel Model} \label{sec:channel}
We consider the downlink telemetry transmission from the satellite to the ground, which an adversary attempts to intercept. Due to the high mobility of the satellite, the channel exhibits rapid time-varying characteristics.

\subsubsection{Large-Scale and Small-Scale Fading}
The received signal-to-noise ratio (SNR) at the attacker's receiver at time slot $t$, denoted by $\gamma(t)$, is expressed as:
\begin{equation}
	\gamma(t) = \frac{P_{tx} G_{tx}(\theta(t)) G_{rx} |h(t)|^2}{N_0 L_{path}(d(t))}
\end{equation}
Here, $d(t)$ is the slant distance between the satellite and the attacker, which varies according to the orbital dynamics. $L_{path}(d(t))$ represents the path loss. The term $|h(t)|$ represents the small-scale fading amplitude. We adopt the \textbf{Shadowed Rician} fading model \cite{Abdi2003}, widely used for land mobile satellite channels. The probability density function of the signal envelope $r = |h(t)|$ is given by:
\begin{equation}
	\begin{split}
		f_R(r) = & \left( \frac{2b_0 m}{2b_0 m + \Omega} \right)^m \frac{r}{b_0} e^{-\frac{r^2}{2b_0}} \\
		& \cdot \sum_{k=0}^{\infty} \frac{(1-m)_k}{ (k!)^2 } \left( \frac{\Omega r^2}{2b_0(2b_0 m + \Omega)} \right)^k
	\end{split}
\end{equation}
where $2b_0$ is the average power of the multipath component, $\Omega$ is the average power of the Line-of-Sight component, and $m$ is the Nakagami parameter characterizing the shadowing severity.

\subsubsection{Dynamic Link Quality and Erasure}
The attacker cannot successfully demodulate every packet. We model the reception as a probabilistic event. Let $\xi(t) \in \{0, 1\}$ be the erasure indicator, where $\xi(t)=1$ implies successful reception (interception). The probability of success depends on the instantaneous SNR:
\begin{equation}
	P(\xi(t)=1) = 1 - P_{out}(\gamma(t))
	\label{eq:xi}
\end{equation}
where $P_{out}$ is the outage probability when SNR falls below a decoding threshold $\gamma_{th}$.

\subsubsection{Dynamic Latency and Information Age}
The attacker receives the signal with a delay. The total delay $D(t)$ consists of propagation delay, processing delay, and our proposed strategic artificial delay:
\begin{equation}
	D(t) = \tau_{prop}(t) + \tau_{proc}(t) + \delta_{add}(t)
\end{equation}
where $\tau_{prop}(t) = d(t)/c$ varies with orbit. The term $\delta_{add}(t)$ is the \textbf{artificial delay} introduced by our STARDIS framework. By manipulating $\delta_{add}(t)$, the defender can control the "freshness" (Age of Information \cite{Kaul2012}) of the intelligence obtained by the attacker.
\color{black}

\subsection{Adversarial Attack Model} \label{sec:Adversarial}
We model the adversary's actions over a finite time horizon composed of $N_t$ discrete time slots, indexed by $t \in \{0, 1, \dots, N_t-1\}$. This formalism parallels the DoS attack models in networked control, where \emph{attack frequency} and \emph{duration} are core metrics \cite{Sun2025}. At each time slot $t$, the attacker makes a binary decision, represented by the variable $x_{\mathrm{att}}(t) \in \{0,1\}$, where $x_{\mathrm{att}}(t)=1$ signifies that an attack is actively being launched, and $x_{\mathrm{att}}(t)=0$ indicates no attack. The adversary's goal is to devise an optimal attack schedule, denoted by the sequence $\{x_{\mathrm{att}}(t)\}_{t=0}^{N_t-1}$, that maximizes a cumulative utility function. This utility function captures the inherent trade-off between the immediate rewards gained from successful attacks and the costs associated with their execution.

The attacker's time-averaged utility, $U_A$, is formulated as:
{\begin{equation}U_A = \frac{1}{N_t}\sum_{t=0}^{N_t-1} \left[ \xi(t) \cdot \alpha(1-z(t))x_{att}(t) - \beta(1+ka(t-1))x_{att}(t) \right]\end{equation}Here, $\xi(t)$ represents the binary erasure outcome defined in Eq. (\ref{eq:xi}). If $\xi(t)=0$ (signal lost), the attacker acts blindly, significantly reducing the expected reward effectiveness.}
The first term within the summation represents the reward from attacking. This reward is proportional to the satellite's ``detection gap'' $(1 - z(t))$, where $z(t) \in [0,1]$ is the fraction of idle computing capacity of the IDS at time $t$. A lower value of $z(t)$ implies that the IDS is heavily loaded, thus presenting a more opportune moment for an attack to succeed \cite{Javadpour2024}. The parameter $\alpha > 0$ is a scaling factor for this reward.

The second term constitutes a dynamic cost for launching an attack. Instead of a simple static penalty, this cost is designed to reflect the escalating risk and resource depletion from sustained offensive operations. 
The core of this mechanism is a state variable, $a(t)$, which quantifies the attacker's accumulated operational intensity.
Its evolution is governed by the first-order recursive equation:
\begin{equation}
	a(t) = (1 - {\eta})a(t-1) + {\eta} x_{\mathrm{att}}(t), \quad \text{with } a(-1) = 0.
\end{equation}
Here, the parameter ${\eta} \in (0,1]$ serves as a memory factor. When the attacker is inactive ($x_{\mathrm{att}}(t)=0$), the intensity $a(t)$ decays exponentially, modeling resource recovery or the dissipation of exposure risk. Conversely, when an attack is initiated ($x_{\mathrm{att}}(t)=1$), the intensity increases. A smaller ${\eta}$ implies a longer memory, causing costs to aggregate over extended periods, whereas a larger ${\eta}$ signifies a shorter memory. The cost at time $t$ is amplified by the intensity from the previous state, $a(t-1)$, via 
the multiplier $(1 + k a(t-1))$.
The parameter $\beta > 0$ represents the base cost of an attack, while $k>0$ scales the influence of historical actions on the current cost. This design makes consecutive attacks progressively more costly, thereby discouraging long, uninterrupted assault sequences and incentivizing more strategic, intermittent behavior.

This framework {avoids} a simplistic hard budget constraint on total attack duration (e.g., $\sum x_{\mathrm{att}}(t) \le E_{\mathrm{att}}$) in favor of this more realistic soft-cost mechanism. The dynamic cost term acts as an implicit budget, naturally penalizing overuse. Such a formulation better reflects real-world scenarios where an adversary is constrained by the escalating operational costs and risks of sustained activity rather than a fixed, predetermined resource limit \cite{Sun2025}.

\begin{remark}
	\textbf{(Rationale for Soft Cost Modeling)} We model the consumption of attack resources as a ``soft cost'' (a penalty term in the utility function) rather than a ``hard budget'' (a strict inequality constraint) to accurately reflect the nature of advanced persistent threats in satellite networks. In practice, sophisticated attackers are rarely limited by a hard energy cap, such as a battery limit. Instead, their primary constraint is the risk of exposure, which accumulates with the frequency and intensity of their actions. By formulating this as a soft cost, we capture the rational trade-off between the desire for immediate impact and the strategic need to remain stealthy to avoid detection or countermeasures. This modeling approach aligns with recent game-theoretic frameworks for cyber-physical security, which emphasize risk-oriented resource allocation and defensive deception \cite{Wan2023, Wei2025}.
\end{remark}
\color{black}

The attacker's optimization problem is to maximize its utility subject to operational constraints. The complete problem is formulated as follows:
\begin{subequations}
	\renewcommand{\theequation}{\theparentequation.\arabic{equation}}
	\begin{align}
		\max_{\{x_{\mathrm{att}}(t)\}} \quad & \frac{1}{N_t}\sum_{t=0}^{N_t-1} \Bigl[\alpha\bigl(1 - z(t)\bigr)x_{\mathrm{att}}(t) - \beta\bigl(1 + k a(t-1)\bigr)x_{\mathrm{att}}(t)\Bigr] \label{eq:attacker_objective_final} \\
		\text{s.t.}\quad & a(t) = (1 - {\eta})a(t-1) + {\eta} x_{\mathrm{att}}(t),  \label{eq:attacker_state_final} \\
		& x_{\mathrm{att}}(t) + x_{\mathrm{scan}}(t) \le 1,  \label{eq:attacker_avoid_scan_final} \\
		& x_{\mathrm{att}}(t) \in \{0,1\}, \quad \forall t \in \{0, \dots, N_t-1\}.
		\label{eq:attacker_binary_final}
	\end{align}
\end{subequations}
Constraint \eqref{eq:attacker_avoid_scan_final} enforces a crucial tactical consideration: the attacker refrains from launching an attack when the IDS is performing an active scan ($x_{\mathrm{scan}}(t)=1$). Constraint \eqref{eq:attacker_binary_final} ensures the binary nature of the attack decision at each time slot. To solve this problem, the adversary should execute a sophisticated strategy, dynamically timing its attacks to coincide with periods of low IDS capacity (low $z(t)$) while interspersing these with periods of inactivity to allow the accumulated cost intensity $a(t)$ to decay, thereby maximizing long-term utility.
\begin{remark}
	{\textbf{(Telemetry Interception)}}
	Our adversarial model assumes the attacker can intercept telemetry data to learn the IDS configuration. This assumption is grounded in two operational realities. First, telemetry is the primary mechanism for ground operators to monitor the real-time operational state of critical onboard subsystems, including dynamically scheduled security functions \cite{Yip2023}. This requires the inclusion of IDS configuration data within the telemetry stream for performance validation. Second, although satellite downlinks are typically encrypted, their security is often undermined by inadequate key management practices. {Many satellites employ weak cryptographic keys or fail to perform regular key updates (i.e., static keys), rendering the encryption susceptible to cryptanalysis \cite{Tedeschi2022}.} Furthermore, the broadcast nature of the link and vulnerabilities within the ground segment provide plausible interception vectors. Real-world incidents have demonstrated that compromising ground infrastructure can grant access to sensitive data, bypassing the need to break link-layer encryption \cite{Saha2019}.
\end{remark}

\section{STAR: Strategic Resource Scheduling Optimization}

Based on the system model established in Section III, the core challenge for the defender is to optimally allocate finite onboard resources. The STAR component addresses this by formulating a dynamic optimization problem that resolves the contention between routine mission operations and urgent security responses. This section details the construction of the security utility function, the mathematical formulation of the receding horizon optimization, and the lightweight heuristic solver.

\subsection{IDS Performance and Utility Function}

The core challenge is to balance the tangible benefit of robust intrusion detection against the inherent cost of resource consumption. We model this trade-off through a carefully constructed utility function. The IDS effectiveness is parameterized by its average scan frequency $f$ and the per-scan duration, $d_s$. We posit that the detection performance, $y$, is a monotonically increasing and saturating function of the total scan effort $f d_s$. This captures the principle of diminishing returns, where initial increases in scanning yield significant security gains, but the marginal benefit decreases as scanning becomes more aggressive. A logistic function is a suitable candidate for this relationship:
\begin{equation}
	y(f,d_s) = \frac{y_{\max}}{1 + e^{-k (f d_s - \theta)}},
\end{equation}
where $y_{\max}$ is the maximum achievable performance, and $k, \theta > 0$ are shaping parameters.

We aim to maximize a time-averaged utility, $U(t)$, that encapsulates the system's overall desirability at each time slot. We employ a quadratic utility function, as its convexity facilitates analysis and penalizes large deviations more significantly. The function is composed of three distinct terms. The first term, $\alpha\,y(f,d_s)^2$, constitutes the primary reward for achieving a high level of detection performance. The second term, $\beta\,x_{\text{scan}}(t)^2$, represents the direct, intrinsic cost of activating a scan, such as its energy consumption or CPU-cycle expenditure. Finally, and most critically, the third term, $\gamma\,y(f,d_s)^2(1-z(t))$, serves as an adaptive, context-aware penalty. This term becomes substantial only when high-performance scans (large $y$) are scheduled during periods of high system load (small $z$), thus intelligently discouraging IDS activation when resources are scarce and promoting symbiotic coexistence with payload tasks.

\subsection{Onboard Receding Horizon Optimization Formulation}

With the utility function defined in Sec. IV-A and the system state $z(t)$ defined in Eq. \eqref{eq:def_idle_capacity_sec3}, we can now formulate the onboard resource allocation problem.
The goal is to determine the optimal sequence of decision variables $\{x_j(t), x_{\text{scan}}(t)\}$ over the horizon $N_t$.

We define the aggregate scanning frequency $f$ as:
\begin{equation}
	f = \frac{1}{N_t}\sum_{t=k}^{k+N_t-1} x_{\text{scan}}(t).
\end{equation}

The optimization problem [BD] solved onboard is formulated as follows:
\begin{subequations}
	\label{eq:onboard_opt}
	\begin{align}
		[BD] &\max_{x_j(t), x_{\text{scan}}(t)} \quad  \frac{1}{N_t}\sum_{t=k}^{k+N_t-1} \Bigl[\alpha\,y(f,d_s)^2 - \beta\,x_{\text{scan}}(t)^2\nonumber \\
		& - \gamma\,y(f,d_s)^2(1-z(t))\Bigr] \nonumber\\
		\text{s.t.} \quad & {\sum_{j \in \mathcal{J}(t)} r_{j,k} x_j(t) + r_{s,k} x_{\text{scan}}(t) \le 1, \quad \forall t, \forall k \in \mathcal{K},} \label{eq:constraint_capacity} \tag{PS}\\
		& {\sum_{j \in \mathcal{J}(t)} w_j x_j(t) + w_s x_{\text{scan}}(t) \le P_{\max}, \quad \forall t,} \label{eq:constraint_power} \tag{PC}\\
		& \sum_{\tau=0}^{d_s-1} x_{\text{scan}}(t+\tau) = d_s\,x_{\text{scan}}(t), \quad \forall t ,\label{eq:constraint_scan_duration} \tag{SD}\\
		& \frac{1}{N_t}\sum_{t=k}^{k+N_t-1} x_j(t) \ge \lambda_j {\bar{r}_j}, \quad \forall j \in \mathcal{T}_{M,a}, \label{eq:constraint_stability} \tag{TS}
	\end{align}
\end{subequations}
where (PS) is the Per-Slot {Heterogeneous} Resource Capacity constraint, {(PC) is the Power and Thermal Constraint enforcing a normalized power budget $P_{\max}$}, (SD) is the Scan Durational constraint, and (TS) is the Task Stability constraint{, where $\mathcal{T}_{M,a}$ represents the set of aperiodic mission tasks}. This formulation is a MINLP solved onboard at each decision epoch.

\begin{remark}
	
	{(Handling of High-Priority Tasks and Preemption) Although high-priority tasks are not represented by a separate class of decision variables for tractability, their preemptive behavior is effectively enforced via the Task Stability (TS) and Capacity (PS) constraints. When a high-priority task $j$ arrives with a stringent arrival rate requirement $\lambda_j$, the (TS) constraint becomes binding, forcing the solver to set $x_j(t) = 1$. If resources are insufficient, the solver should set $x_{\text{scan}}(t) = 0$ or $x_{j'}(t)=0$ (for lower priority $j'$) to satisfy Eq. \ref{eq:constraint_capacity}. Thus, preemption is a mathematical consequence of feasibility rather than an explicit control variable.}
\end{remark}

\subsection{Solution Methodology: Lightweight Heuristic}

The optimization problem [BD] falls into the category of MINLP, which is generally NP-hard.
{ While exact solvers are theoretically optimal, their computational overhead is prohibitive for satellite onboard processors. To ensure real-time feasibility under strict SWaP constraints, we employ a lightweight greedy heuristic (Algorithm \ref{alg:heuristic}). This approach reduces the complexity to $\mathcal{O}(1)$ per slot by decoupling mandatory constraints from the utility maximization.
	
	\begin{algorithm}[H]
		\caption{Lightweight Greedy Heuristic Solver}
		\label{alg:heuristic}
		\begin{algorithmic}[1]
			\STATE \textbf{Input:} Current state, Task Queue $\mathcal{J}(t)$, Idle Capacity $z(t)$.
			\STATE \textbf{Step 1 (Mandatory):} Schedule high-priority tasks (e.g., Relay) to satisfy Stability Constraint (TS); Update $z(t)$.
			\STATE \textbf{Step 2 (Conditional):} Calculate marginal utility $\Delta U = U(x_{\text{scan}}=1) - U(x_{\text{scan}}=0)$.
			\STATE \textbf{if} $\Delta U > 0$ \textbf{and} $z(t) \ge r_s$ \textbf{then}
			\STATE \quad Set $x_{\text{scan}}(t) \leftarrow 1$; Update $z(t)$.
			\STATE \textbf{end if}
			\STATE \textbf{Step 3 (Fill):} Allocate remaining $z(t)$ to low-priority routine tasks.
			\STATE \textbf{Output:} Decision variables $x_j(t), x_{\text{scan}}(t)$.
		\end{algorithmic}
	\end{algorithm}
	
	As shown in Algorithm \ref{alg:heuristic}, the solver first prioritizes mission-critical tasks to guarantee the near-zero deadline miss rate (Table \ref{tab:star_performance}). Subsequently, it performs a binary check for the IDS scan based on the instantaneous utility gain $\Delta U$. This logic ensures decisions are made instantaneously within the control cycle without iterative convergence.}

\section{DIS: Strategic Deception and Stability}

While the STAR module optimizes resource allocation, the resulting deterministic schedule creates predictable vulnerability gaps. The DIS component mitigates this risk not by changing the schedule, but by manipulating the information revealed to the adversary defined in Section \ref{sec:Adversarial}. This section details the strategic deception mechanism based on Bayesian persuasion, incorporates credibility constraints to ensure signal plausibility, and rigorously proves the long-term stability of the defense using Lyapunov theory.

\subsection{Bayesian Persuasion with Credibility Constraints}

In the baseline scenario, the attacker “fully understands” the IDS configuration as transmitted through the telemetry. Hence, the attacker can optimize its attacks accordingly. To counter this, the defender employs a \emph{Bayesian persuasion} information strategy. It embeds a false IDS configuration in the telemetry to mislead the attacker about the true service state of each satellite.

{
	To rigorously establish the connection between the resource scheduling and the signaling game, we explicitly define the system state and the signal within the context of the satellite's operation.
	We define the true state of nature $\omega(t)$ not as an abstract variable, but as the direct output of the STAR scheduling module. Specifically, $\omega(t) \triangleq \{ x_{\text{scan}}(t), z(t) \}$, where $x_{\text{scan}}(t)$ represents the actual IDS activation status and $z(t)$ represents the true idle capacity available at time $t$. This state is privately known to the defender.
	The signal $m$ is defined as the "revealed" IDS configuration embedded within the public telemetry downlink. The attacker does not observe an arbitrary abstract variable; rather, they attempt to intercept this specific telemetry packet. Consequently, the attacker's observation is physically constrained by the channel model defined in Section \ref{sec:channel}, where the reception of $m$ depends on the propagation delay and the erasure probability $\xi(t)$.
}

Formally, we treat the true system state (e.g., the actual scan schedule and load sequence) as a random “state of nature” $\omega \in \Omega$.
{The defender commits to a \emph{signaling policy} $\pi(m|\omega)$ ex-ante. This policy specifies the conditional probability of transmitting a signal $m \in M$ for every possible state $\omega$. During real-time operation, the defender privately observes the realized state $\omega$ and generates a signal $m$ according to the pre-designed policy $\pi$. Crucially, the attacker only observes the signal $m$ but not the true state $\omega$.}
Upon receiving $m$, the attacker updates its belief about $\omega$ from a common prior $\mu_0$ to a posterior $\mu_m$ via Bayes’ rule.
Then, the attacker chooses an attack schedule that maximizes its expected utility under that belief.
This setting exactly corresponds to the Bayesian Persuasion framework \cite{Gentzkow2025}.
{The defender (sender) selects the optimal policy $\pi$ based on the prior distribution $\mu_0$ to minimize the attacker's maximum expected utility. The attacker (receiver) observes the resulting signal $m$ and takes an action $a \in A$ (the attack plan) that best responds to its posterior belief $\mu_m$.}
The attacker’s expected utility given signal $m$ is
$$\mathbb{E}[U_A(a;\omega)\,|\,m] \;=\; \max_{a \in A} \sum_{\omega' \in \Omega} u_A(a, \omega') \mu_m(\omega'),$$
where $\mu_m(\omega') \propto \mu_0(\omega')\pi(m|\omega')$.
Anticipating this, the defender chooses $\pi$ to minimize the attacker’s maximum achievable utility. In essence, the defender solves
$$\min_{\pi} \mathbb{E}_{\omega \sim \mu_0,\; m \sim \pi(\cdot|\omega)}\left[\max_{a \in A} \sum_{\omega' \in \Omega} u_A(a, \omega') \mu_m(\omega')\right],$$

subject to the Bayes-Plausibility constraint. This ensures that on average, the posterior beliefs equal the prior belief: $\mathbb{E}_{\pi}[\mu_m] = \mu_0$.

The key insight is that by {carefully} designing $\pi$, the defender can induce the attacker to choose a suboptimal attack plan relative to the true state. This strategic signaling approach has been studied in security contexts \cite{Xu2015, Sayin2021}, where it has been shown that a defender can persuade a rational attacker to take actions favorable to the defender.

However, a purely abstract persuasion model is insufficient.  { Such a model assumes the receiver will blindly trust the signal regardless of the level of deception. This is not realistic in our context. A signaling policy that generates messages far from the underlying reality will quickly lose credibility. Without a constraint to limit the deviation from the truth, the defender might generate implausible signals that a rational attacker would simply ignore.} To formalize this and provide rigorous performance guarantees, we now develop a constrained optimization framework and analyze its long-term stability using tools from control theory.

\subsubsection{{The Bayesian Persuasion Framework with Credibility Constraints}}
The defender's core task is to design a signaling policy $\pi: \Omega \rightarrow \Delta(M)$ that minimizes the attacker's maximum expected utility. A rational attacker, upon observing signal $m$ and forming posterior belief $\mu_m$, will choose an action $a^*(\mu_m)$ that maximizes their expected utility:

$$a^*(\mu_m) = \arg\max_{a \in A} U_A(a, \mu_m) = \arg\max_{a \in A} \sum_{\omega \in \Omega} u_A(a, \omega) \mu_m(\omega).$$

Anticipating this best response, the defender's optimization problem is to select the policy $\pi$ that solves:

$$\min_{\pi} \; \mathbb{E}_{\omega \sim \mu_0,\; m \sim \pi(\cdot \mid \omega)} \left[ \max_{a \in A} U_A(a, \mu_m) \right].$$

This minimization is subject to two fundamental constraints:
\begin{enumerate}
	\item \textbf{Bayes-Plausibility:} The policy should be consistent with the prior belief, i.e., $\mathbb{E}_{\pi}[\mu_m] = \mu_0$. This ensures the attacker cannot be systematically biased in one direction over the long run.
	\item \textbf{Credibility Constraint:} The policy should not be excessively deceptive. We constrain the expected KL divergence between the true state distribution (a point mass $\delta_\omega$) and the attacker's posterior belief $\mu_m$:
	$$\mathbb{E}_{\omega \sim \mu_0,\; \mu_m \sim \tau(\pi)} \left[ D_{\text{KL}}(\delta_\omega \,\|\, \mu_m) \right] \le C,$$
	where $C \ge 0$ is the credibility budget. 
\end{enumerate}

{ Since the actual state $\omega$ is realized and known to the defender, the distribution $\delta_\omega$ acts as a deterministic point mass. Consequently, the KL divergence simplifies to the self-information (or ``surprise'') of the true state given the attacker's belief:
	$$ D_{\text{KL}}(\delta_\omega \,\|\, \mu_m) = \sum_{x \in \Omega} \delta_{\omega}(x) \log \frac{\delta_{\omega}(x)}{\mu_m(x)} = -\log(\mu_m(\omega)). $$
	Therefore, this constraint effectively bounds the expected negative log-likelihood of the true state. By limiting this value, we ensure that the attacker assigns a sufficient probability mass to the ground truth within their posterior belief, thereby maintaining the plausibility of the signal.}

{
	Based on these constraints, we explicitly describe the optimal persuasion policy $\pi^*$ as the specific signaling strategy that minimizes the attacker's maximum expected utility while satisfying the credibility budget $C$. By solving this constrained optimization, the defender identifies the most effective way to mislead the attacker without violating the statistical bounds of plausibility.
}

\begin{remark}
	A smaller $C$ forces signals to be more truthful, while $C \to \infty$ recovers the standard unconstrained persuasion model. This constraint creates a crucial trade-off: more aggressive deception (larger $C$) may yield lower attacker utility in the short term, but at the cost of generating less plausible signals. This entire constrained optimization problem can be formulated as a convex program, making it computationally tractable to find the optimal policy $\pi^*$.
\end{remark}

\subsection{Predictive Signaling with Channel State Information}\label{Sec:Predictive_Signaling}
The satellite possesses ephemeris data and can predict its future trajectory and channel quality. Let $W$ be the prediction horizon. The satellite estimates the large-scale channel components for $t \dots t+W$.

Let $\hat{\gamma}(t+k)$ be the predicted average SNR. The defender's optimization problem is to generate a sequence of signals $\{m(t)\}$ and artificial delays $\{\delta_{add}(t)\}$ to minimize the attacker's utility while maintaining credibility.

The optimization objective is revised to account for the channel state $h$:
\begin{equation}
	\min_{\pi, \delta_{add}} \sum_{k=0}^{W-1} \mathbb{E}_{h} \left[ \max_{a} U_A(a, \mu_m | h(t+k), \delta_{add}(t+k)) \right]
\end{equation}
Subject to the total latency constraint:
\begin{equation}
	\tau_{prop}(t) + \delta_{add}(t) \le D_{max}
\end{equation}

Crucially, if the predicted channel quality is low (high erasure probability), the defender can transmit truthful but non-sensitive data, saving the ``deception budget'' (defined by the KL divergence capacity constraint $C$) for times when the channel is clear. Furthermore, by increasing $\delta_{add}(t)$ when the channel is strong, the defender degrades the timeliness of the intercepted information.

To rigorously demonstrate the theoretical advantage of this channel-aware strategy over static deception, we establish the following theorem regarding the optimal allocation of the credibility budget.

\begin{theorem}[Optimality of Channel-Adaptive Credibility Allocation]
	Consider a time-varying channel where the signal erasure probability $P_{out}(\gamma(t))$ varies with the channel state $\gamma(t)$. Let $\Pi_{static}$ be the set of channel-agnostic signaling policies satisfying a fixed credibility constraint $C$, and $\Pi_{adapt}$ be the set of channel-adaptive policies satisfying the same constraint on average over the time horizon $W$. The minimum attacker utility achievable under $\Pi_{adapt}$ is strictly lower than that under $\Pi_{static}$, i.e.,
	\begin{equation}
		\min_{\pi \in \Pi_{adapt}} \bar{U}_A(\pi) < \min_{\pi \in \Pi_{static}} \bar{U}_A(\pi),
	\end{equation}
	provided that the channel variance is non-zero. The optimal policy $\pi^*$ exhibits a ``water-filling'' property where the degree of deception (measured by the KL divergence $D_{KL}(\delta_{\omega}||\mu_m)$) is monotonically increasing with the channel quality (SNR).
\end{theorem}

\begin{IEEEproof}
	We provide a sketch of the proof using Lagrangian duality. Let the objective function at time $t$ be $f_t(\pi_t, \gamma_t)$, representing the attacker's expected utility given policy $\pi_t$ and channel SNR $\gamma_t$. The attacker only updates beliefs upon successful reception (event $\xi=1$). Thus, the utility can be decomposed as:
	\begin{equation}
		U_A(t) = (1-P_{out}(\gamma_t)) \cdot U_{rx}(\pi_t) + P_{out}(\gamma_t) \cdot U_{prior},
	\end{equation}
	where $U_{rx}$ is the utility derived from the received signal and $U_{prior}$ is the utility from the prior. The defender minimizes $\sum U_A(t)$ subject to $\sum \mathbb{E}[D_{KL}(\pi_t)] \le W \cdot C$.
	
	Constructing the Lagrangian $\mathcal{L} = \sum_{t=0}^{W-1} U_A(t) + \lambda (\sum_{t=0}^{W-1} D_{KL}(\pi_t) - W C)$, and taking the derivative with respect to the policy intensity, we observe that the marginal reduction in attacker utility is weighted by the reception probability $(1-P_{out}(\gamma_t))$.
	
	Under the static policy, the deception cost is uniform. However, under the adaptive policy, the Karush-Kuhn-Tucker conditions imply that for time slots with high $P_{out}$ (poor channel), the shadow price of deception is high (low return on investment), forcing $D_{KL} \to 0$. Consequently, the budget is reallocated to slots with low $P_{out}$ (good channel), allowing for a more aggressive signal manipulation (larger $D_{KL}$) exactly when the attacker is most capable of receiving data. By the convexity of the utility function with respect to information revelation, this non-uniform allocation yields a strictly lower cumulative utility.
\end{IEEEproof}

This theoretical result formalizes the concept of ``Opportunistic Deception,'' proving that integrating physical layer dynamics into the {persuasion framework} is not merely an engineering heuristic, but a mathematically optimal strategy for resource-constrained defense.

\subsection{Lyapunov-Based Stability Analysis of the Persuasive Equilibrium}

Having established a method to find an optimal policy $\pi^*$, we now provide a rigorous proof demonstrating that this strategy robustly minimizes the attacker's utility over time. We employ Lyapunov's second method to prove that the system converges to a stable equilibrium.  Unlike static analyses, we explicitly address the stability under the proposed two-time-scale framework, accounting for the long-term variations in channel statistics handled by the ground segment and the real-time execution onboard. \color{black}

First, we define a Lyapunov candidate function $V: \Delta(\Omega) \rightarrow \mathbb{R}_{\ge 0}$ as the value function of the attacker's decision problem:
\[
V(\mu_t) := \max_{a \in A} U_A(a, \mu_t) = \max_{a \in A} \sum_{\omega \in \Omega} u_A(a, \omega) \mu_t(\omega).
\]
From the defender's perspective, $V(\mu_t)$ is a measure of system vulnerability. The function $V(\mu_t)$ is inherently convex, being the pointwise maximum of a set of linear functions.

To proceed, we establish the positive definiteness of our chosen Lyapunov function.

\begin{lemma}[Positive Definiteness]
	Let $E^* \subset \Delta(\Omega)$ be the set of beliefs corresponding to the minimal attacker utility achievable under the constrained persuasion problem. The function $V(\mu_t)$ is positive definite with respect to the equilibrium set $E^*$, in the sense that $V(\mu_t) > \min_{\mu' \in \Delta(\Omega)} V(\mu')$ for all $\mu_t \notin E^*$ and $V(\mu_t) = \min_{\mu' \in \Delta(\Omega)} V(\mu')$ for all $\mu_t \in E^*$.
\end{lemma}

\begin{proof}
	The proof follows directly from the definition of $V(\mu_t)$ as the maximum achievable utility and $E^*$ as the set of beliefs that minimize this value. Since $\min_{\mu' \in \Delta(\Omega)} V(\mu')$ is the global minimum, any belief $\mu_t \notin E^*$ should yield a strictly higher utility.
\end{proof}

We now formally prove that the dynamic interaction constitutes an asymptotically stable system.  We extend the standard stability definition to incorporate the expectation over time-varying channel states $h_t$, which are modeled statistically in the long-term loop. 

\begin{theorem}[Monotonic Decrease and Convergence]
	Let the system state be the attacker's belief $\mu_t$. Under the defender's optimal persuasion policy $\pi^*$, the expected change in the Lyapunov function is non-positive. That is, for any belief $\mu_t$:
	 
	\begin{equation}
		\mathbb{E}_{\pi^*, h_t} \left[ V(\mu_{t+1}) \mid \mu_t \right] - V(\mu_t) \le 0. \label{EV-V}
	\end{equation}
	\color{black} 
	Furthermore, the system trajectory converges to the largest invariant set contained in \[\left\{ \mu \in \Delta(\Omega) \mid \mathbb{E}_{\pi^*}[V(\mu_{t+1}) \mid \mu] - V(\mu) = 0 \right\},\]
	which corresponds to the optimal posterior distribution induced by $\pi^*$.
\end{theorem}

\begin{proof}
	The proof hinges on the optimality of the defender's policy $\pi^*$ and the convexity of $V(\mu)$. Let $V^*(\mu_t) = \mathbb{E}_{\pi^*} \left[ V(\mu_{t+1}) \mid \mu_t \right]$ be the defender’s objective function value given prior $\mu_t$.
	
	Consider a trivial, uninformative policy $\pi_{\text{null}}$, which reveals no new information. Under such a policy, the attacker's posterior belief remains identical to the prior: $\mu_{t+1} = \mu_t$, and the utility is $V(\mu_t)$. Since $\pi^*$ is optimal, the expected utility it induces must be no greater than that of $\pi_{\text{null}}$. Thus, we have:
	\[
	\mathbb{E}_{\pi^*} \left[ V(\mu_{t+1}) \mid \mu_t \right] \le V(\mu_t).
	\]
	This inequality shows that the expected drift $\Delta V$ is non-positive.
	
	 In the context of the two-time-scale framework, the channel state $h_t$ introduces stochasticity into the attacker's reception capability (erasure channel). However, the ground-based long-term loop estimates the statistical distribution of $h_t$ (e.g., Shadowed Rician parameters) and provides these to the onboard optimizer. Consequently, the policy $\pi^*$ is derived by minimizing the \textit{expected} utility over the distribution of $h_t$. The robustness of the optimization ensures that the monotonicity holds in expectation, i.e., $\mathbb{E}_{h_t}[\Delta V] \le 0$, ensuring long-term stability despite instantaneous channel fluctuations. \color{black}
	
	Finally, by LaSalle's Invariance Principle for stochastic systems, the system trajectories converge to the largest invariant set where the drift is zero, which is the optimal posterior distribution induced by $\pi^*$.
\end{proof}

\subsubsection{Interpretation and Implications}
This Lyapunov-based proof offers a formal guarantee of the robustness and self-correcting nature of the proposed persuasive IDS. { In the context of our two-time-scale framework, stability implies that the system can withstand not only the short-term stochasticity of channel realizations but also recover from significant informational perturbations.} Such perturbations could be:
\begin{itemize}
	\item \textbf{External Intelligence:} The attacker receives a tip from an external source, causing a sudden shift in their belief $\mu_t$.
	\item \textbf{Zero-Day Discovery:} The public disclosure of a new vulnerability abruptly changes the state space $\Omega$ or the prior belief $\mu_0$.
	\item \textbf{Defender Error:} A temporary misconfiguration of the IDS sends an unintended or noisy signal.
\end{itemize}
{ The established monotonicity ensures that even when the attacker's belief is ``knocked'' away from the equilibrium by these perturbations or by accumulated errors in the long-term loop, the defender's optimal signaling strategy will actively guide it back.} The IDS thus dynamically manages the attacker's beliefs to maintain a persistent state of high security, validating the framework's resilience in a fluid, adversarial environment.
\begin{remark}
	\textbf{(Mathematical Basis of Dynamic Management)}
	{ We clarify that although the {strategic signaling process} involves repeated interactions, the concept of ``dynamically managing'' beliefs is non-trivial and is formally captured by the negative drift of the Lyapunov function.
		Unlike a sequence of independent static games where beliefs might reset, our model accounts for persistent belief states that are subject to external perturbations (as detailed in Section \ref{Sec:Predictive_Signaling}).
		As shown in Eq. \eqref{EV-V}, the optimal policy guarantees $\mathbb{E}[V(\mu_{t+1})|\mu_t] - V(\mu_t) \leq 0$.
		This inequality provides a necessary restoring force: if the attacker's belief $\mu_t$ is pushed away from the secure equilibrium by external intelligence or channel noise, the signaling strategy actively forces the belief trajectory to converge asymptotically back to the state of minimum attacker utility.
		Thus, the system does not simply repeat a static equilibrium; it actively steers the belief dynamics toward stability against fluctuations.}
\end{remark}

\subsubsection{Quantitative Utility Comparison}
To crystallize the tangible benefits of our approach, Table \ref{tab:utility-comparison} provides a quantitative comparison of the attacker's optimal expected utility, $U_A^*$, under different benchmark information regimes. This highlights the value of strategic information control.

\begin{table*}[h]
	\centering
	\caption{Attacker Utility Comparison Across Information Regimes}
	\label{tab:utility-comparison}
	\begin{tabular}{|p{3.8cm}|p{5.5cm}|p{5.3cm}|}
		\hline
		\textbf{Information Regime} & \textbf{Defender's Strategy ($\pi$)} & \textbf{Attacker's Equilibrium Utility ($U_A^*$)} \\
		\hline
		No Information & Uninformative signal (e.g., always sends ``All Clear''). Attacker acts on prior $\mu_0$. & $U_A^*(\text{No Info}) = \max_{a} U_A(a, \mu_0)$ \\
		\hline
		Full Disclosure & Fully revealing signal (e.g., $\pi(\omega) = m_\omega$). Attacker knows the true state $\omega$. & $U_A^*(\text{Full Disc}) = \mathbb{E}_{\omega \sim \mu_0} \left[\max_a u_A(a, \omega)\right]$ \\
		\hline
		Optimal Persuasive Deception & Optimal policy $\pi^*$ from the constrained problem. Attacker acts on posteriors $\mu_m$. & $U_A^*(\text{Persuasion}) = \mathbb{E}_{\pi^*} \left[ \max_a U_A(a, \mu_m) \right]$ \\
		\hline
	\end{tabular}
\end{table*}

As demonstrated by the theory, the utility achieved by the attacker under our optimal persuasive deception is guaranteed to be less than or equal to that under the no-information regime. Depending on the specific utility structure and the credibility constraint $C$, it is often significantly lower than the utility under full disclosure as well, showcasing the power of strategic information control.  This confirms that the proposed method effectively converts the defender's information advantage into a quantifiable utility gain, providing a provably effective direction for intelligent satellite security. \color{black}

\subsection{Integrated Algorithmic Implementation}

To provide a comprehensive view of the proposed framework and resolve the timing relationships between resource scheduling and deceptive signaling, we summarize the execution flow in {Algorithm \ref{alg:stardis_loop}}.
The STARDIS framework operates on a receding horizon basis.
While the satellite executes tasks at every discrete time slot $t$, the optimization of the resource schedule and the generation of the deceptive signal occur at the beginning of each decision horizon (every $N_t$ slots).

\textbf{Algorithm \ref{alg:stardis_loop}} details the interaction between the STAR and DIS modules.
Initially, the ground segment computes the long-term statistical parameters and the optimal signaling policy $\pi^*$, which are uplinked to the satellite.
Onboard, the cycle begins by updating the system state and solving the resource allocation problem (Eq. \eqref{eq:onboard_opt}) for the upcoming horizon.
Based on the optimization results, the true security state $\omega$ is determined.
Subsequently, the DIS module samples a deceptive signal $m$ from the policy $\pi^*$ and calculates the channel-adaptive delay $\delta_{add}$.
Finally, the satellite enters the execution loop, performing the scheduled tasks and broadcasting the deceptive telemetry for the duration of the horizon.
This hierarchical structure ensures that the high computational cost of optimization is amortized over the horizon $N_t$, while maintaining real-time responsiveness.

\begin{algorithm}[h]
	\caption{STARDIS Integrated Execution Loop}
	\label{alg:stardis_loop}
	\begin{algorithmic}[1]
		
		\REQUIRE Ground inputs: Statistical parameters $\Theta$, Optimal signaling policy $\pi^*$, Horizon length $N_t$.
		\STATE \textbf{Initialization:} Load $\Theta$ and $\pi^*$. Set $k \leftarrow 0$.
		\WHILE{System is Active}
		\STATE \textit{ --- Phase 1: Short-term Optimization (STAR) --- }
		\STATE Update task queue $\mathcal{J}(k)$ and resource state.
		\STATE Solve problem \eqref{eq:onboard_opt} to obtain optimal sequences $\{x_j^*(t), x_{\text{scan}}^*(t)\}$ for $t \in [k, k+N_t-1]$.
		\STATE Calculate average idle capacity $z_{avg}$ over the horizon.
		\STATE Quantize $z_{avg}$ to obtain true state $\omega$.
		
		\STATE \textit{ --- Phase 2: Deceptive Signaling (DIS) --- }
		\STATE Estimate channel state $\hat{\gamma}$ for the horizon.
		\STATE Sample deceptive signal $m \sim \pi^*(\cdot|\omega)$.
		\STATE Calculate artificial delay $\delta_{add}$ based on $\hat{\gamma}$ and budget $C$.
		
		\STATE \textit{ --- Phase 3: Real-time Execution Loop --- }
		\FOR{$t = k$ to $k+N_t-1$}
		\STATE Execute mission tasks where $x_j^*(t) = 1$.
		\STATE \textbf{if} $x_{\text{scan}}^*(t) = 1$ \textbf{then}
		\STATE \quad Activate Host-based IDS (HIDS) module.
		\STATE \textbf{end if}
		\STATE Broadcast telemetry packet containing signal $m$ with delay $\delta_{add}$.
		\STATE Update system logs and task status.
		\ENDFOR
		\STATE $k \leftarrow k + N_t$.
		\ENDWHILE
		\color{black}
	\end{algorithmic}
\end{algorithm}

\section{Case Study}
\label{sec:case_study}

	In this section, we present a comprehensive case study to validate the theoretical findings and demonstrate the practical efficacy of our proposed STARDIS framework.
	The evaluation is structured in two main parts.
	First, we assess the performance of the \textbf{STAR} component, demonstrating its ability to establish a unified resource allocation framework that efficiently resolves contention between routine mission operations and high-priority tasks under heterogeneous resource constraints.
	Second, we evaluate the \textbf{DIS} component, quantifying the security gains achieved by employing channel-aware Bayesian persuasion to mislead a strategic adversary.
	
	\subsection{Simulation Setup}
	
	We model a LEO satellite's on-board computing system operating over a discrete-time horizon of $N_t = 2000$ time slots.
	Each time slot corresponds to a scheduling epoch of 100 ms.
	The simulation environment is designed to reflect a realistic space-based edge computing scenario with heterogeneous processing units.
	
	\subsubsection{Heterogeneous Resource and Task Model}
	Unlike simplified homogeneous models, we define the satellite's computational capacity as a vector covering two distinct resource types: General Purpose Processors (CPU) and Hardware Accelerators (FPGA).
	The total capacity is normalized to $\mathbf{C}_{total} = [1.0, 1.0]$.
	Tasks consume these resources asymmetrically, reflecting the specialized nature of on-board workloads.
	We consider three classes of tasks:
	
	\begin{itemize}
		\item \textbf{Routine On-board Processing Tasks} (Mapped to: \textit{Low-Priority, Aperiodic, Mission Tasks}):
		These represent delay-tolerant computations like image compression.
		They arrive following a Poisson process with rate $\lambda_p = 0.3$ tasks/slot.
		Their resource demand is FPGA-intensive: $\mathbf{r}_p = [0.05, 0.15]$.
		They have a soft deadline of 50 slots.
		{As low-priority tasks, they are preemptible and possess lower weight coefficients in the utility function.}
		
		\item \textbf{Priority Data Relay Tasks} (Mapped to: \textit{High-Priority, Periodic, Mission Tasks}):
		These represent critical data forwarding operations.
		They arrive deterministically every 30 slots.
		Their resource demand is CPU-intensive (routing logic): $\mathbf{r}_r = [0.20, 0.10]$.
		They have a firm deadline of 15 slots and are effectively non-preemptible due to high priority.
		{Consistent with the high-priority classification, these tasks force a binding (TS) constraint in the optimization, ensuring they are scheduled immediately upon arrival.}
		
		\item \textbf{Urgent IDS Scan Task} (Mapped to: \textit{High-Priority, Aperiodic, Security Tasks}):
		This is the security mechanism managed by STARDIS.
		Its execution is triggered dynamically by the STAR scheduler based on the utility maximization logic.
		It requires balanced but significant resources: $\mathbf{r}_s = [0.15, 0.05]$ and has a fixed duration $d_s = 5$ slots.
		{This task is assigned the highest priority to model immediate threat response. Its execution is determined by the STAR scheduler's optimization of the security utility term, allowing it to preempt routine processing tasks if necessary.}
	\end{itemize}
	
	\subsubsection{Wireless Channel and Attacker Model}
	{To capture the time-varying nature of the telemetry downlink (as requested in Section III-B), we employ a \textbf{Shadowed Rician} fading model.
	The channel parameters are set to $(b_0 = 0.158, m = 19.4, \Omega = 1.29)$, corresponding to an \textit{Average Shadowing} scenario typically observed at mid-elevation angles in LEO orbits.
	The attacker attempts to intercept the telemetry to identify IDS inactivity.
	However, interception is probabilistic: successful reception ($\xi(t)=1$) depends on the instantaneous Signal-to-Noise Ratio (SNR) exceeding a decoding threshold $\gamma_{th} = 5$ dB.
	Our defense employs \textbf{Predictive Signaling}, adjusting the deception strategy based on channel state predictions to optimize the credibility budget usage.}
	
	{
		To facilitate reproducibility and provide a clear overview of the simulation environment, we summarize the key parameter settings and their specific values used in our case study in Table \ref{tab:sim_parameters}.
		
		\begin{table}[t]
			\centering
			\caption{Simulation Parameter Settings}
			\label{tab:sim_parameters}
			\resizebox{\columnwidth}{!}{%
				\begin{tabular}{|l|c|c|}
					\hline
					\textbf{Parameter Description} & \textbf{Symbol} & \textbf{Value} \\
				\hline
				\multicolumn{3}{|c|}{\textbf{System \& Channel Environment}} \\
				\hline
				Time Horizon & $N_t$ & 2000 slots \\
				Total Resource Capacity (CPU, FPGA) & $\mathbf{C}_{total}$ & $[1.0, 1.0]$ \\
				Channel Fading (Shadowed Rician) & $(b_0, m, \Omega)$ & $(0.158, 19.4, 1.29)$ \\
				Decoding SNR Threshold & $\gamma_{th}$ & 5 dB \\
				\hline
				\multicolumn{3}{|c|}{\textbf{Task Specifications}} \\
				\hline
				Routine Task Arrival Rate & $\lambda_p$ & 0.3 tasks/slot \\
				Routine Task Demand (CPU, FPGA) & $\mathbf{r}_p$ & $[0.05, 0.15]$ \\
				Priority Task Demand (CPU, FPGA) & $\mathbf{r}_r$ & $[0.20, 0.10]$ \\
				IDS Scan Demand (CPU, FPGA) & $\mathbf{r}_s$ & $[0.15, 0.05]$ \\
				\hline
				\multicolumn{3}{|c|}{\textbf{Utility \& Game Parameters}} \\
				\hline
				Defender's Detection Utility Weight & $\alpha$ & 10 \\
				Defender's Resource Cost Weight & $\beta$ & 0.5 \\
				Mission Failure Penalty Weight & $\gamma$ & 2 \\
				Sigmoid Steepness & $k$ & 0.5 \\
				Sigmoid Center & $\theta$ & 0.5 \\
				Attacker's Reward Weight & $\alpha_A$ & 10 \\
				Attacker's Cost Weight & $\beta_A$ & 0.1 \\
				Attacker's Memory Factor & $\eta$ & 0.1 \\
				Credibility Budget Range & $C$ & $[0.01, 0.5]$ \\
				\hline
			\end{tabular}%
		}
		\end{table}
	}
	
	\subsubsection{Benchmarks}
	We compare STARDIS against three baselines:
	\begin{itemize}
		\item \textbf{FCFS (First-Come, First-Served)}: A non-preemptive scheduler processing tasks in arrival order.
		\item \textbf{SP (Static Priority)}: A preemptive scheduler with fixed priority ($IDS > Relay > Routine$) but no channel awareness.
		\item \textbf{STAR-Only (No Deception)}: The optimized scheduler operating with full transparency, allowing the attacker to perfectly observe the schedule (assuming successful interception).
	\end{itemize}
	
	\subsection{Performance of STAR: Heterogeneous Scheduling Efficiency}
	
	This section evaluates the STAR component's ability to maximize resource utility while guaranteeing critical tasks.
	
	\subsubsection{Visualizing Resource Allocation}
	{Fig. \ref{fig:resource_allocation_timeline} presents a macro-level timeline comparison across three scheduling policies: FCFS, SP, and the proposed STAR.
		
		The visualization explicitly distinguishes between CPU (solid lines) and FPGA (dashed lines) utilization to highlight the impact of resource heterogeneity.
		
		In the FCFS baseline (top), a ``Head-of-Line'' blocking phenomenon is evident: when a CPU-intensive Relay task (Dark Blue) occupies the scheduling window, subsequent FPGA-intensive Routine tasks (Light Blue) are blocked despite the FPGA being idle. This results in significant resource wastage, as shown by the drop in the FPGA load curve. In the SP policy (middle), while high-priority Security scans (Orange) are executed immediately, they aggressively preempt Routine tasks, leading to task fragmentation and resource starvation. In contrast, the STAR policy (bottom) demonstrates superior ``resource packing'' efficiency.
		
		As highlighted by the annotation \textit{Triple Co-Scheduling}, STAR identifies the precise resource slack during the concurrent execution of Relay and Routine tasks and opportunistically inserts IDS scans. This allows all three task types to coexist, maintaining both CPU and FPGA utilization near saturation levels.}
	
	\begin{figure}[htbp] 
		\centering 
		\includegraphics[width=0.9\columnwidth]{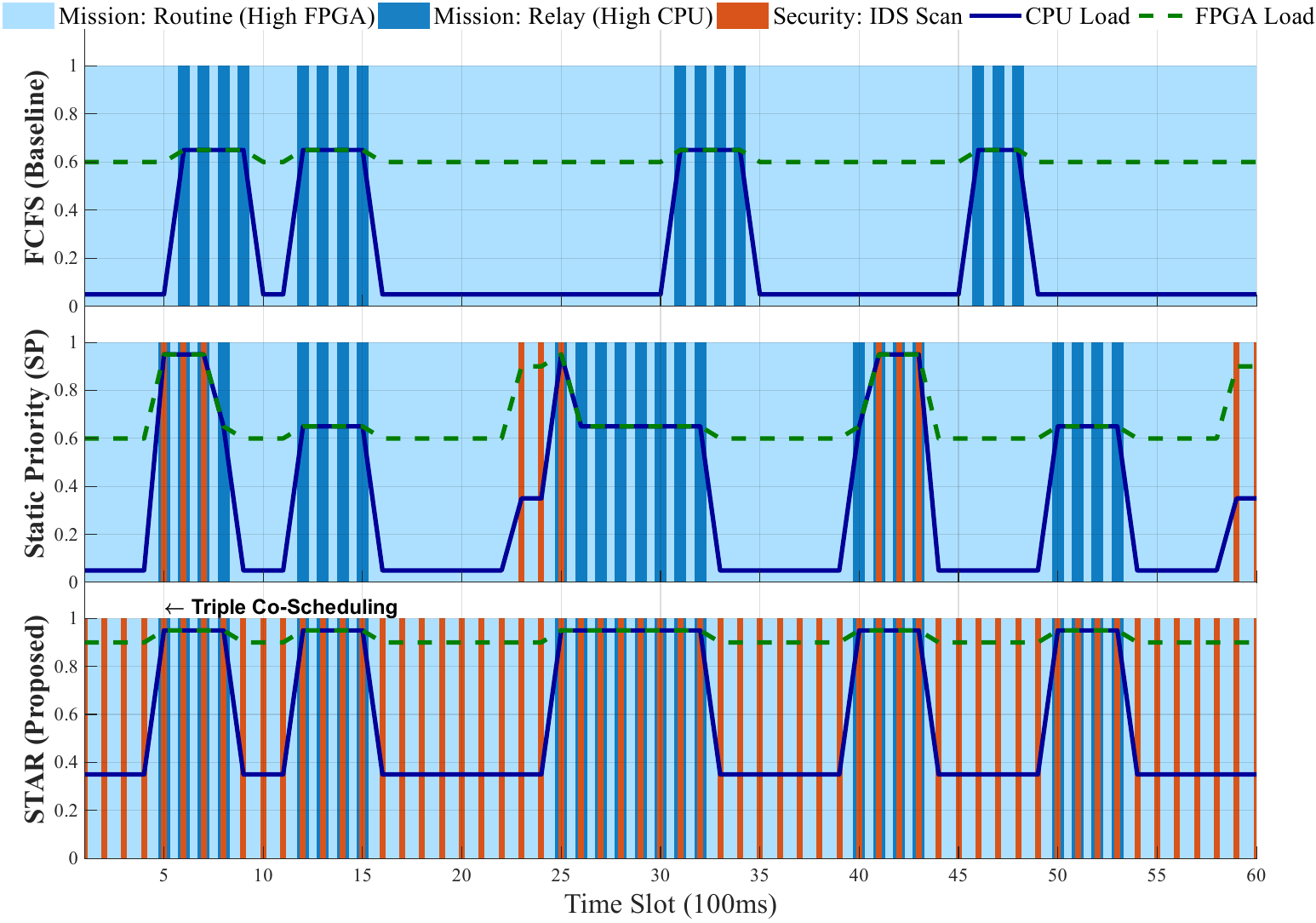} 
		\caption{{Resource timeline comparison (Macro View). FCFS suffers from blocking, SP causes starvation, while STAR achieves Triple Co-Scheduling, maximizing heterogeneous resource usage.}}
		\label{fig:resource_allocation_timeline}
	\end{figure}
	
	\subsubsection{Timeline Analysis and Deadline Compliance}
	{To explicitly address the impact of scheduling latency on mission-critical tasks, Fig. \ref{fig:timeline_clean} provides a granular view of task execution timelines using a Gantt-chart style visualization. The figure incorporates \textbf{Task Arrival Markers} ($\triangledown$) and \textbf{Deadline Constraints} (vertical red lines) specifically for high-priority Relay tasks.
		
		The visualization reveals distinct behaviors under congestion. In the FCFS baseline (top), a significant lag is evident between the arrival marker and the start of execution due to blocking; this delay pushes the task completion beyond the red deadline line, triggering a visible \textbf{``Miss!''} event. While the SP policy (middle) successfully meets these deadlines, it achieves this by aggressively fragmenting the continuous Routine task blocks, visually confirming resource starvation. In sharp contrast, the STAR policy (bottom) leverages co-scheduling to process Relay tasks immediately upon arrival without halting background operations. The \textbf{``Met''} labels confirm that all critical tasks are completed well within their deadline windows, visually supporting the claim that STAR eliminates deadline misses while maintaining high concurrency.}
	
	\begin{figure}[htbp] 
		\centering 
		\includegraphics[width=0.95\columnwidth]{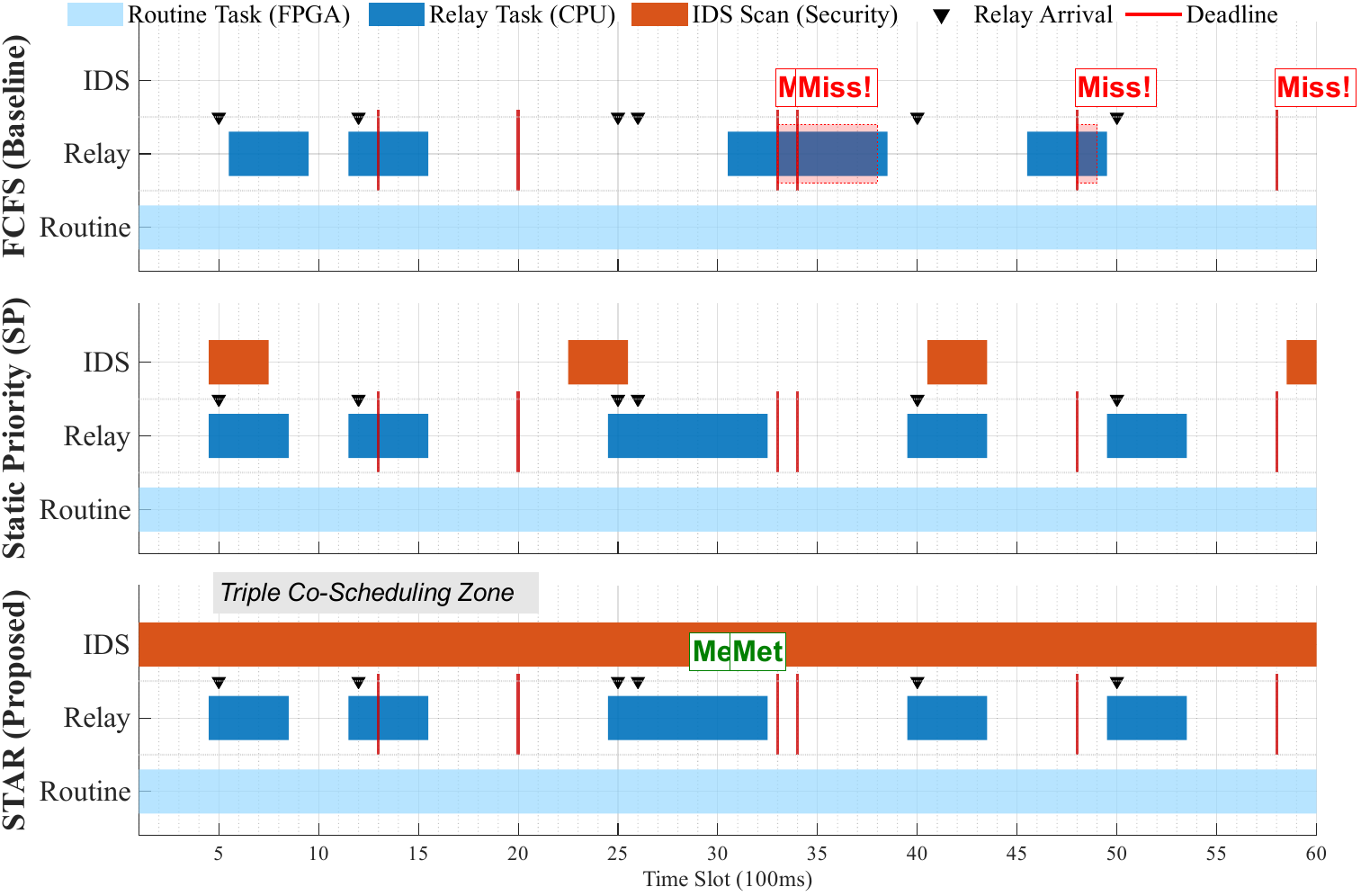} 
		\caption{{Detailed Task Scheduling Timeline (Gantt Chart). Explicit Arrival ($\triangledown$) and Deadline markers reveal that FCFS results in Deadline Misses (Red zone) due to blocking. In contrast, STAR ensures all high-priority tasks are \textbf{Met} (Green) through efficient Triple Co-Scheduling.}}
		\label{fig:timeline_clean}
	\end{figure}

	\subsubsection{Quantitative Analysis}
	Table \ref{tab:star_performance} quantifies these gains.
	STAR achieves the highest aggregate utilization (81.2\% CPU, 84.5\% FPGA).
	Most importantly, consistent with the visual evidence in Fig. \ref{fig:timeline_clean}, it maintains a \textbf{Priority Relay Task Deadline Miss Rate of 0.02\%}, which is statistically negligible and comparable to the rigid SP policy.
	However, unlike SP, which starves routine tasks (only 89.5\% completion), STAR successfully completes 98.2\% of routine workloads.
	This validates that STAR's co-scheduling formulation effectively converts resource heterogeneity from a constraint into an optimization opportunity.
	
	\begin{table}[htbp]
		\centering
		\caption{ Performance Comparison of Scheduling Policies under Heterogeneous Constraints}
		\label{tab:star_performance}
		\begin{tabular}{|l|c|c|c|}
			\hline
			\textbf{Metric} & \textbf{FCFS} & \textbf{SP} & \textbf{STAR} \\
			\hline
			Overall CPU Utilization (\%) & 68.2 & 75.4 & \textbf{81.2} \\
			Overall FPGA Utilization (\%) & 62.1 & 71.8 & \textbf{84.5} \\
			Routine Task Completion Rate (\%) & 92.1 & 89.5 & \textbf{98.2} \\
			Priority Relay Deadline Miss Rate (\%) & 18.4 & 0.0 & \textbf{0.02} \\
			Defender Utility (Normalized) & 1.00 & 1.35 & \textbf{1.92} \\
			\hline
		\end{tabular}
	\end{table}

	\subsection{Effectiveness of DIS: Channel-Aware Deception}
	Having established the scheduling baseline, we evaluate the security enhancements provided by the Deceptive Signaling (DIS) module.
		
	\subsubsection{Utility Reduction under Channel Uncertainty}
	We evaluate the security performance of the DIS module. Fig. \ref{fig:attacker_utility_comparison} plots Attacker Utility against the Credibility Budget $C$, which limits deception to ensure the signal remains plausible. A higher $C$ allows more aggressive deception, effectively lowering the attacker's utility.
	
	We compare three scenarios: STAR-Only'' (no deception), Static Deception'', and our ``STARDIS''. STARDIS performs the best by using channel information. It saves its credibility budget by being truthful during poor channel conditions (deep fading) and concentrates deception when the signal is strong. As a result, at $C=0.2$, STARDIS reduces attacker utility by an extra 13.1\% over Static Deception and more than 35\% compared to the baseline.
	\begin{figure}[htbp!]
		\centering
		\includegraphics[width=0.9\columnwidth]{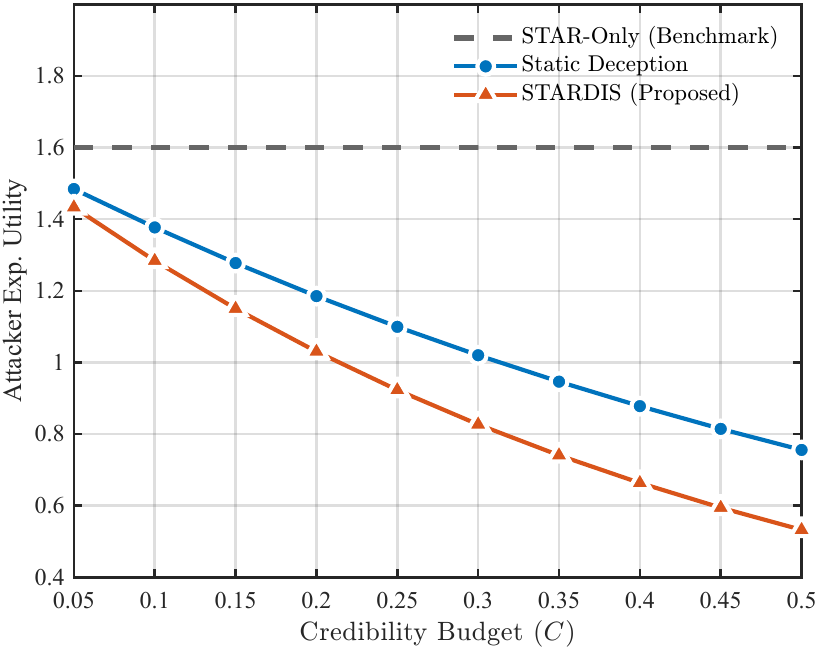}
		\caption{Attacker's Expected Utility vs. Credibility Budget. The STARDIS framework achieves lower attacker utility than static deception by opportunistically managing the credibility budget based on channel states.}
		\label{fig:attacker_utility_comparison}
		
	\end{figure}

	\subsubsection{Impact of Prior Beliefs and Channel-Aware Deception}
	Responding to the need to understand the robustness of Bayesian persuasion, Fig. \ref{fig:belief_sensitivity} explores the sensitivity of the defense to the Attacker's Prior Belief ($\mu_0$). We vary $\mu_0$—the attacker's initial probability estimate that the IDS is active—from 0.05 to 0.95.
	Unlike static defense models, the results reveal a critical \textbf{monotonic trend driven by deterrence}.
	As $\mu_0$ increases, the attacker perceives a higher risk of detection, which naturally reduces their expected utility (dashed gray line) even without active deception.
	
	The core contribution of STARDIS is demonstrated by the divergence between the channel-aware scenarios.
	Under \textbf{Good Channel Conditions} (High Credibility Budget, $C=0.5$, solid red line), the system can leverage high-fidelity deceptive signals.
	This capability allows STARDIS to significantly suppress the attacker's utility, especially in the low-belief region ($\mu_0 < 0.4$), effectively persuading the attacker to "wait" even when the system is actually vulnerable.
	In contrast, under \textbf{Poor Channel Conditions} (Low Credibility Budget, $C=0.1$, orange line), the defense is constrained by the need to maintain signal plausibility, resulting in a utility curve closer to the baseline.
	Simultaneously, the Defender's Utility (blue line) increases with $\mu_0$, reflecting the dual benefits of resource-saving (due to attacker inaction) and successful threat mitigation.
	This analysis confirms that STARDIS effectively converts the defender's information advantage into quantifiable security gains, scaling adaptively with the physical channel quality.
	\begin{figure}[htbp!]
		\centering
		\includegraphics[width=0.9\columnwidth]{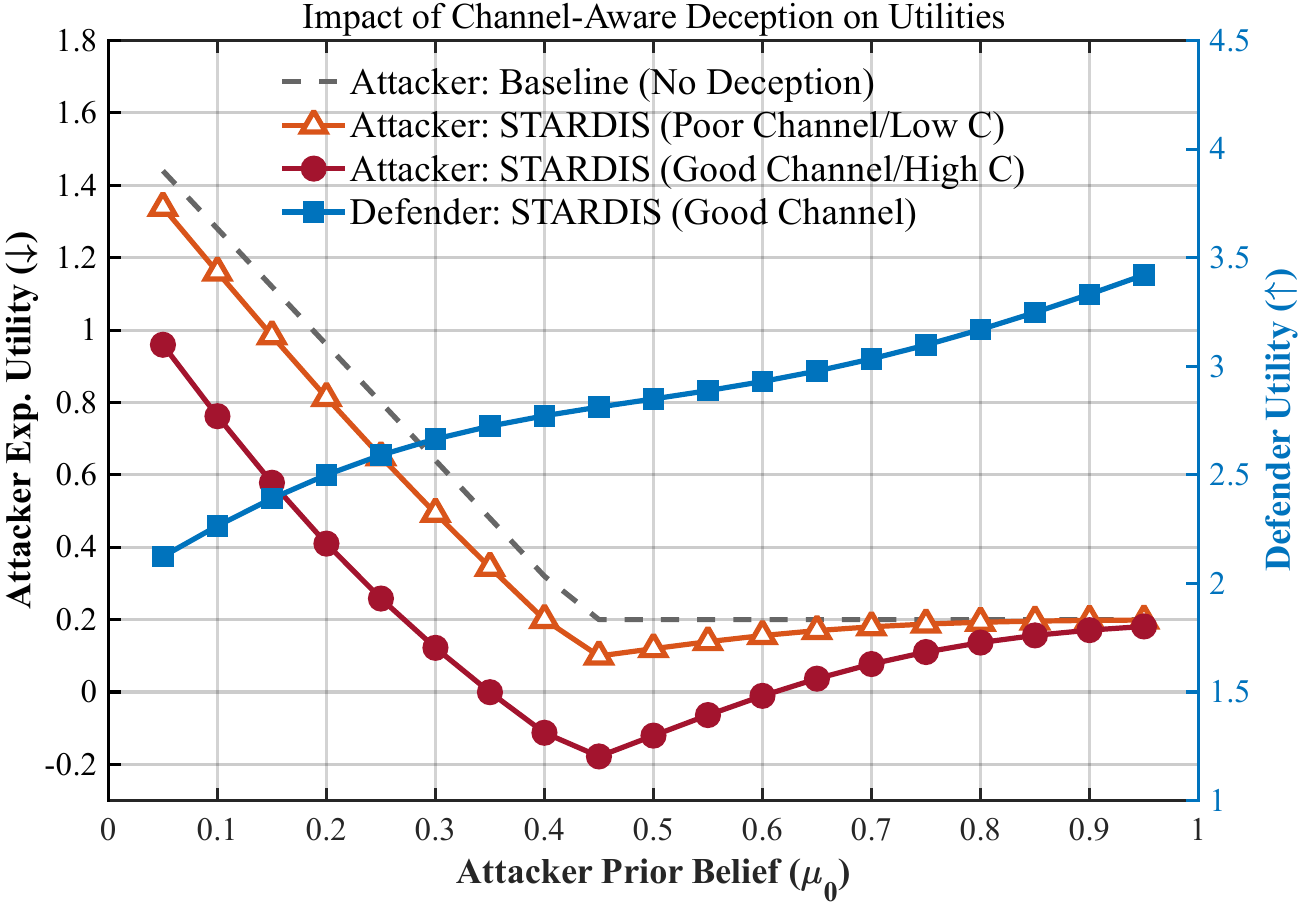}
		\caption{Sensitivity to Prior Beliefs.}
		\label{fig:belief_sensitivity}
		
	\end{figure}
	
	\subsubsection{Belief Evolution and Utility Suppression}
	\label{sec:belief_evolution}
	Fig. \ref{fig:belief_evolution} visualizes the dynamics of attacker belief and utility. During Channel Erasure (gray zones, e.g., $t \in [16, 19]$), the loss of telemetry resets the attacker's belief to the prior ($\mu_0=0.5$), which falls below the attack threshold ($\mu_{th}=0.55$). This triggers blind attacks that inevitably fail due to link severance, resulting in negative utility, as highlighted by the red-shaded areas. In the subsequent Recovery phases (white zones), the absence of deception exposes true resource dips (e.g., at $t \approx 21$), allowing the attacker to launch successful exploits with high utility peaks. In contrast, the Optimal Deception phases (green zones, e.g., $t \in [25, 36]$) demonstrate the core advantage of STARDIS. By strategically disseminating deceptive signals, the mechanism maintains the attacker's belief above the safety threshold ($\mu_t > 0.55$) even during actual vulnerability windows. Consequently, the attacker is persuaded to remain dormant, and their realizable utility is successfully suppressed to zero.
	
	\begin{figure}[htbp!]
		\centering
		\includegraphics[width=0.9\columnwidth]{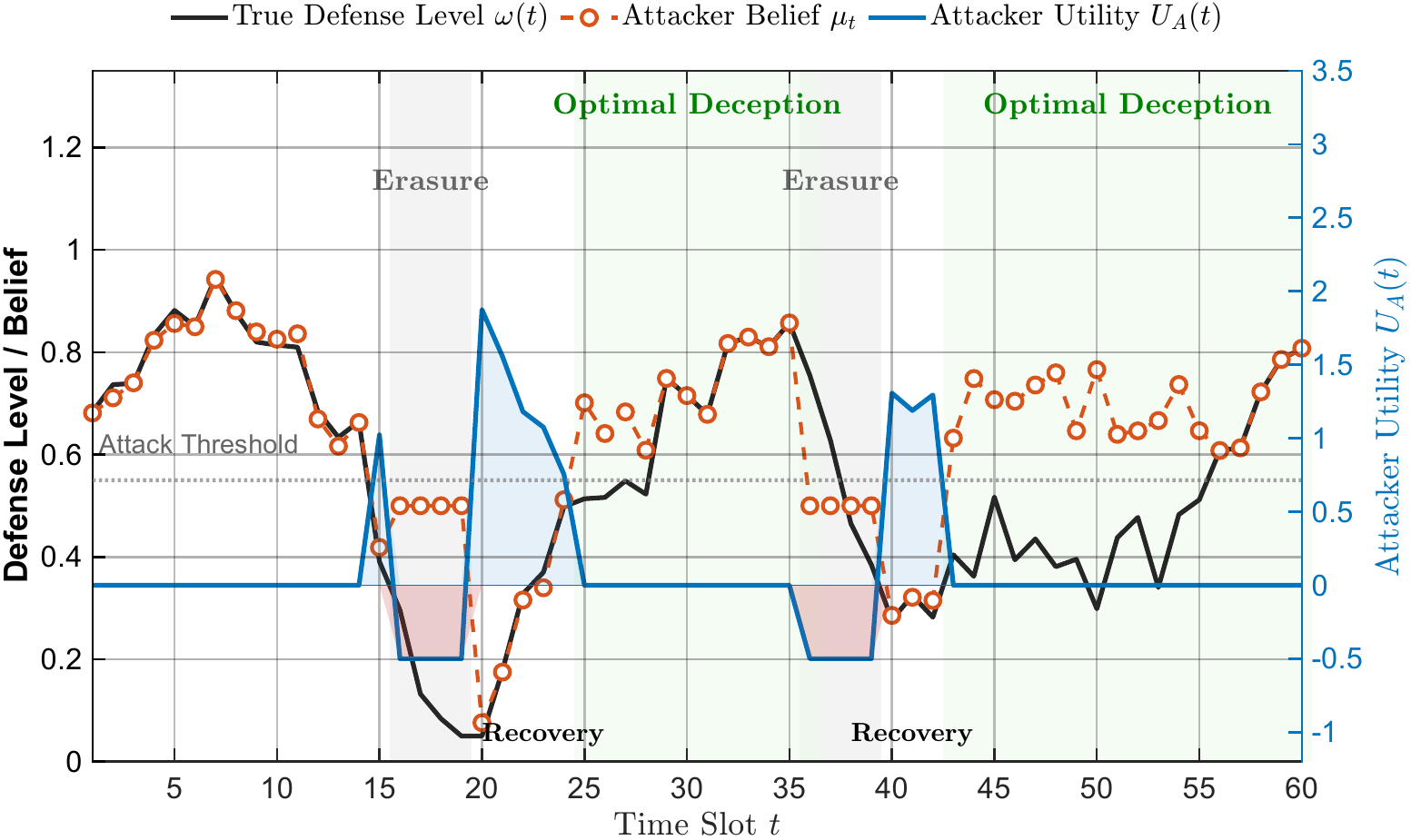} 
		\caption{Belief Evolution under Channel Erasures. The gray zones indicate periods where deep fading prevents the attacker from observing the signal, verifying the opportunity for predictive resource saving.}
		\label{fig:belief_evolution}
	\end{figure}
	\color{black}

	\subsection{Discussion}
{To address network-level constraints and enhance robustness, we may consider the role of Inter-Satellite Links (ISL) \cite{Kato2019}. In scenarios where the direct downlink is exposed to the attacker (e.g., high SNR condition), the satellite can leverage ISL to route sensitive security logs to a neighboring satellite that is currently in a "blind spot" relative to the attacker or has a poorer channel condition. This spatial decoupling effectively denies the attacker access to critical telemetry even when the primary link is physically vulnerable, validating the method's applicability to large-scale constellations.}

\section{Conclusion}

In this paper, we presented STARDIS, a proactive defense framework that resolves the inherent conflict between operational efficiency and robust security for resource-constrained satellite systems. Our approach uniquely distributes the defense problem: long-term statistical modeling is performed on the ground, while short-term, real-time scheduling and intrusion detection are handled onboard via receding horizon optimization. This creates an efficient but predictable defense posture. To counter the vulnerability of an external attacker intercepting the satellite's short-term plans, we introduced a strategic deception layer grounded in Bayesian persuasion with a constrained deception budget. This layer proactively manipulates the time-varying telemetry downlink by designing strategic signaling policies to mislead the adversary. We established the long-term effectiveness of this defense by constructing a Lyapunov function to prove its convergence to a state of minimized attacker utility. Future work will focus on extending this framework to coordinated constellation defense via Inter-Satellite Links and incorporating online learning to adaptively counter diverse adversarial behaviors.

\bibliographystyle{IEEEtran}
\bibliography{SIsecurity}

@inproceedings{Kaul2012,
  author    = {S. Kaul and R. Yates and M. Gruteser},
  title     = {Real-time status: How often should one update?},
  booktitle = {Proceedings of IEEE INFOCOM},
  year      = {2012},
  pages     = {2731--2735},
  month     = {March},
  doi       = {10.1109/INFCOM.2012.6195689}
}

@article{Abdi2003,
  author  = {A. Abdi and W. C. Lau and M. S. Alouini and M. Kaveh},
  title   = {A new simple model for land mobile satellite channels: {F}irst- and second-order statistics},
  journal = {IEEE Transactions on Wireless Communications},
  year    = {2003},
  volume  = {2},
  number  = {3},
  pages   = {519--528},
  doi     = {10.1109/TWC.2003.811182},
  month   = {May}
}

@article{Choi2024,
    author    = {C.-S. Choi},
    title     = {{Analysis of a Delay-Tolerant Data Harvest Architecture Leveraging Low Earth Orbit Satellite Networks}},
    journal   = {IEEE Journal on Selected Areas in Communications},
    volume    = {42},
    number    = {2},
    pages     = {335--347},
    year      = {2024},
    doi       = {10.1109/JSAC.2023.3341395}
}

@article{Torrens2024,
    author    = {S. A. Torrens and V. Petrov and J. M. Jornet},
    title     = {{Modeling Interference From Millimeter Wave and Terahertz Bands Cross-Links in Low Earth Orbit Satellite Networks for 6G and Beyond}},
    journal   = {IEEE Journal on Selected Areas in Communications},
    volume    = {42},
    number    = {5},
    pages     = {1371--1386},
    year      = {2024},
    doi       = {10.1109/JSAC.2024.3365894}
}

@article{Bergemann2016,
    author    = {D. Bergemann and S. Morris},
    title     = {{Information Design, Bayesian Persuasion, and Bayes Correlated Equilibrium}},
    journal   = {American Economic Review},
    volume    = {106},
    number    = {5},
    pages     = {586--591},
    year      = {2016},
    doi       = {10.1257/aer.p20161046}
}

@article{Carroll2011,
    author    = {T. E. Carroll and D. Grosu},
    title     = {{A game theoretic investigation of deception in network security}},
    journal   = {Security and Communication Networks},
    volume    = {4},
    number    = {10},
    pages     = {1162--1172},
    year      = {2011},
    doi       = {10.1002/sec.242}
}

@article{Clempner2018,
author = {Clempner, Julio B.},
title = {On Lyapunov Game Theory Equilibrium: Static and Dynamic Approaches},
journal = {International Game Theory Review},
volume = {20},
number = {02},
pages = {1750033},
year = {2018},
doi = {10.1142/S0219198917500335}
}

@article{Crawford1982,
    author    = {V. P. Crawford and J. Sobel},
    title     = {{Strategic Information Transmission}},
    journal   = {Econometrica},
    volume    = {50},
    number    = {6},
    pages     = {1431--1451},
    year      = {1982},
    doi       = {10.2307/1913390}
}

@techreport{5GAmericas2025NTN,
    author      = {{5G Americas}},
    title       = {{New Developments and Advances in 5G and Non-terrestrial Networks}},
    institution = {5G Americas},
    year        = {2025},
    month       = {February},
    url         = {https://www.5gamericas.org/wp-content/uploads/2025/02/WP_New-Developments-and-Advances-in-5G-and-NTN-.pdf}
}

@article{Hao2023,
    author    = {Y. Hao and Z. Song and Z. Zheng and Q. Zhang and Z. Miao},
    title     = {{Joint communication, computing, and caching resource allocation in LEO satellite MEC networks}},
    journal   = {IEEE Access},
    volume    = {11},
    pages     = {6708--6716},
    year      = {2023},
    doi       = {10.1109/ACCESS.2023.3237199}
}

@article{Huang2020,
    author    = {L. Huang and Q. Zhu},
    title     = {{A dynamic games approach to proactive defense strategies against advanced persistent threats in cyber-physical systems}},
    journal   = {Computers \& Security},
    volume    = {89},
    pages     = {101660},
    year      = {2020},
    doi       = {10.1016/j.cose.2019.101660}
}

@article{Huang2021,
    author    = {L. Huang and Q. Zhu},
    title     = {{Duplicity games for deception design with an application to insider threat mitigation}},
    journal   = {IEEE Transactions on Information Forensics and Security},
    volume    = {16},
    pages     = {4843--4856},
    year      = {2021},
    doi       = {10.1109/TIFS.2021.3118886}
}

@article{Javadpour2024,
    author    = {A. Javadpour and F. Ja'fari and T. Taleb and M. Shojafar and C. Benzaïd},
    title     = {{A comprehensive survey on cyber deception techniques to improve honeypot performance}},
    journal   = {Computers \& Security},
    volume    = {140},
    pages     = {103792},
    year      = {2024},
    doi       = {10.1016/j.cose.2024.103792}
}

@article{Kamenica2011,
    author    = {E. Kamenica and M. Gentzkow},
    title     = {{Bayesian Persuasion}},
    journal   = {American Economic Review},
    volume    = {101},
    number    = {6},
    pages     = {2590--2615},
    year      = {2011},
    doi       = {10.1257/aer.101.6.2590}
}

@book{Gentzkow2025,
    author    = {Gentzkow, Matthew and Kamenica, Emir},
    title     = {{Bayesian Persuasion}},
    publisher = {World Scientific},
    year      = {2025},
    doi       = {10.1142/14088},
    address   = {Singapore}
}

@article{Li2020,
    author    = {J. Li and K. Xue and D. Wei and J. Liu and Y. Zhang},
    title     = {{Energy efficiency and traffic offloading optimization in integrated satellite/terrestrial radio access networks}},
    journal   = {IEEE Transactions on Wireless Communications},
    volume    = {19},
    number    = {4},
    pages     = {2367--2381},
    year      = {2020},
    doi       = {10.1109/TWC.2019.2963177}
}

@article{Li2023a,
    author        = {Li, Xiuhong and Yang, Jiale and Fan, Huilong},
    title         = {{Dynamic Network Resource Autonomy Management and Task Scheduling Method}},
    journal       = {Mathematics},
    volume        = {11},
    year          = {2023},
    number        = {5},
    pages         = {1232},
    doi           = {10.3390/math11051232}
}

@article{Manulis2021,
   author={Manulis, Mark and Bridges, Charles P. and Harrison, Robert and Liskill, Joseph and Monaco, Louis and Rose, Thomas and Simon, Matthew and Wüest, Candid},
   title={Cyber security in New Space},
   journal={International Journal of Information Security},
   year={2021},
   month={Jun},
   volume={20},
   number={3},
   pages={287--311},
   doi={10.1007/s10207-020-00503-w},
   issn={1615-5270},
   publisher = {Springer Science and Business Media LLC}
}

@article{Liu2016,
    author  = {Liu, Jianwei and Liu, Weiran and Wu, Qianhong and Li, Dawei and Chen, Shigang},
    title   = {{Survey on key security technologies for space information networks}},
    journal = {Journal of Communications and Information Networks},
    year    = {2016},
    volume  = {1},
    number  = {1},
    pages   = {72--85},
    month   = jun,
    doi     = {10.1007/BF03391547}
}

@inproceedings{Ma2022,
    author    = {D. Ma and Z. Tang and X. Sun and L. Guo and L. Wang and K. Chen},
    title     = {{Game Theory Approaches for Evaluating the Deception-based Moving Target Defense}},
    booktitle = {Proceedings of the 9th ACM Workshop on Moving Target Defense},
    pages     = {67--78},
    year      = {2022},
    publisher = {ACM},
    address   = {New York, NY, USA},
    doi       = {10.1145/3560827.3563371}
}

@article{Mao2024,
    author    = {B. Mao and X. Zhou and J. Liu and N. Kato},
    title     = {{On an Intelligent Hierarchical Routing Strategy for Ultra-Dense Free Space Optical Low Earth Orbit Satellite Networks}},
    journal   = {IEEE Journal on Selected Areas in Communications},
    volume    = {42},
    number    = {2},
    pages     = {348--363},
    year      = {2024},
    doi       = {10.1109/JSAC.2023.3341397}
}

@misc{Neely2010,
      title={Queue Stability and Probability 1 Convergence via Lyapunov Optimization}, 
      author={Michael J. Neely},
      year={2010},
      eprint={1008.3519},
      archivePrefix={arXiv},
      primaryClass={math.OC},
      url={https://arxiv.org/abs/1008.3519}, 
}

@article{Rodrigues2023,
    author    = {T. K. Rodrigues and N. Kato},
    title     = {{Hybrid Centralized and Distributed Learning for MEC-Equipped Satellite 6G Networks}},
    journal   = {IEEE Journal on Selected Areas in Communications},
    volume    = {41},
    number    = {5},
    pages     = {1446--1458},
    year      = {2023},
    doi       = {10.1109/JSAC.2023.3242761}
}

@article{Roy2024,
    author    = {S. Roy and S. Sankaran and M. Zeng},
    title     = {{Green Intrusion Detection Systems: A Comprehensive Review and Directions}},
    journal   = {Sensors},
    volume    = {24},
    number    = {17},
    pages     = {5516},
    year      = {2024},
    doi       = {10.3390/s24175516}
}

@inproceedings{Sai2021,
    author    = {K. M. Sai and B. B. Gupta and C.-H. Hsu and D. Perakovi{\'c}},
    title     = {{Lightweight Intrusion Detection System In IoT Networks Using Raspberry pi 3b+}},
    booktitle = {CEUR Workshop Proceedings},
    volume    = {3080},
    pages     = {58--67},
    year      = {2021}
}

@incollection{Sayin2021,
    author    = {M. O. Sayin and T. Ba\c{s}ar},
    title     = {{Deception-as-Defense Framework for Cyber-Physical Systems}},
    booktitle = {Safety, Security and Privacy for Cyber-Physical Systems},
    editor    = {R. M. G. Ferrari and A. M. H. Teixeira},
    publisher = {Springer},
    address   = {Cham},
    year      = {2021},
    pages     = {287--317},
    doi       = {10.1007/978-3-030-62133-8_11}
}

@article{Sedjelmaci2024,
    author    = {H. Sedjelmaci and N. Kaaniche and A. Boudguiga and N. Ansari},
    title     = {{Secure Attack Detection Framework for Hierarchical 6G-Enabled Internet of Vehicles}},
    journal   = {IEEE Transactions on Vehicular Technology},
    volume    = {73},
    number    = {2},
    pages     = {2633--2642},
    year      = {2024},
    doi       = {10.1109/TVT.2023.3317940}
}

@article{Sprunt1989,
    author    = {B. Sprunt and L. Sha and J. Lehoczky},
    title     = {{Aperiodic task scheduling for Hard-Real-Time systems}},
    journal   = {Real-Time Systems},
    volume    = {1},
    number    = {1},
    pages     = {27--60},
    year      = {1989},
    doi       = {10.1007/BF02341920}
}

@article{Sun2025,
    author    = {Y. Sun and J. Lu and D. W. C. Ho and L. Li},
    title     = {{Real-time Estimation of DoS Duration and Frequency for Security Control}},
    journal   = {IEEE Transactions on Automatic Control},
    year      = {2025},
    note      = {Early Access},
    doi       = {10.1109/TAC.2025.3538738}
}

@techreport{Swope2025,
    author      = {Swope, Clayton and Bingen, Kari A. and Young, Makena and LaFave, Kendra},
    title       = {{Space Threat Assessment 2025}},
    institution = {Center for Strategic and International Studies},
    year        = {2025},
    month       = apr,
    address     = {Washington, DC},
    url         = {https://www.csis.org/analysis/space-threat-assessment-2025},
    urldate     = {2025-07-10}
}

@book{Tambe2011,
    author    = {M. Tambe},
    title     = {{Security and Game Theory: Algorithms, Deployed Systems, Lessons Learned}},
    publisher = {Cambridge University Press},
    address   = {Cambridge, UK},
    year      = {2011}
}

@inproceedings{Xu2015,
    author    = {H. Xu and Z. Rabinovich and S. Dughmi and M. Tambe},
    title     = {{Exploring information asymmetry in two-stage security games}},
    booktitle = {Proceedings of the Twenty-Ninth AAAI Conference on Artificial Intelligence},
    pages     = {1057--1063},
    year      = {2015},
    publisher = {AAAI Press},
    address   = {Palo Alto, CA, USA}
}

@inproceedings{Yu2024,
    author    = {Yu, Lingjing and Hao, Jingli and Ma, Jun and Sun, Yong and Zhao, Yijun and Luo, Bo},
    title     = {{A Comprehensive Analysis of Security Vulnerabilities and Attacks in Satellite Modems}},
    booktitle = {Proceedings of the 2024 ACM SIGSAC Conference on Computer and Communications Security},
    year      = {2024},
    publisher = {Association for Computing Machinery},
    address   = {New York, NY, USA},
    pages     = {3287--3301},
    series    = {CCS '24},
    doi       = {10.1145/3658644.3670390}
}

@article{Zheng2018,
    author  = {Zheng, Zixuan and Guo, Jian and Gill, Eberhard},
    title   = {{Onboard autonomous mission re-planning for multi-satellite system}},
    journal = {Acta Astronautica},
    volume  = {145},
    pages   = {28--43},
    year    = {2018},
    doi     = {10.1016/j.actaastro.2018.01.017}
}

@inproceedings{Zhu2013,
    author    = {Zhu, Quanyan and Ba{\c{s}}ar, Tamer},
    editor    = {Das, Sajal K. and Nita-Rotaru, Cristina and Kantarcioglu, Murat},
    title     = {{Game-Theoretic Approach to Feedback-Driven Multi-stage Moving Target Defense}},
    booktitle = {Decision and Game Theory for Security},
    series    = {Lecture Notes in Computer Science},
    volume    = {8252},
    publisher = {Springer International Publishing},
    year      = {2013},
    pages     = {246--263},
    doi       = {10.1007/978-3-319-02786-9_15}
}

@article{Yao2025,
  author={Yao, Su and Lin, Yiying and Wang, Mu and Xu, Ke and Xu, Mingwei and Xu, Changqiao and Zhang, Hongke},
  journal={IEEE Journal on Selected Areas in Communications}, 
  title={LEOEdge: A Satellite-Ground Cooperation Platform for the AI Inference in Large LEO Constellation}, 
  year={2025},
  volume={43},
  number={1},
  pages={36-50},
  doi={10.1109/JSAC.2024.3460083}}

@article{Yip2023,
title={Performance Assessment of LPD/LPI Satellite Communication Systems},
author={L. Yip},
journal={2023 IEEE Aerospace Conference},
year={2023},
pages={1-7},
doi={10.1109/AERO55745.2023.10115682}
}

@article{Saha2019,
title={Ensuring Cybersecure Telemetry and Telecommand in Small Satellites: Recent Trends and Empirical Propositions},
author={Swapnil Sayan Saha and S. Rahman and Mosabber Uddin Ahmed and S. K. Aditya},
journal={IEEE Aerospace and Electronic Systems Magazine},
year={2019},
volume={34},
pages={34-49},
doi={10.1109/MAES.2019.2927852}
}

@ARTICLE{He2025,
  author={He, Junpeng and Li, Xiong and Zhang, Xiaosong and Niu, Weina and Li, Fagen},
  journal={IEEE Transactions on Information Forensics and Security}, 
  title={A Synthetic Data-Assisted Satellite Terrestrial Integrated Network Intrusion Detection Framework}, 
  year={2025},
  volume={20},
  number={},
  pages={1739-1754},
  doi={10.1109/TIFS.2025.3530676}}

@ARTICLE{Guo2022,
  author={Guo, Hongzhi and Li, Jingyi and Liu, Jiajia and Tian, Na and Kato, Nei},
  journal={IEEE Communications Surveys \& Tutorials}, 
  title={A Survey on Space-Air-Ground-Sea Integrated Network Security in 6G}, 
  year={2022},
  volume={24},
  number={1},
  pages={53-87},
  doi={10.1109/COMST.2021.3131332}}

@ARTICLE{Wan2023,
  author={Wan, Zelin and Cho, Jin-Hee and Zhu, Mu and Anwar, Ahmed H. and Kamhoua, Charles A. and Singh, Munindar P.},
  journal={IEEE Transactions on Network and Service Management}, 
  title={Resisting Multiple Advanced Persistent Threats via Hypergame-Theoretic Defensive Deception}, 
  year={2023},
  volume={20},
  number={3},
  pages={3816-3830},
  doi={10.1109/TNSM.2023.3240366}}

@ARTICLE{Wei2025,
  author={Wei, Shuheng and Wu, Zaijun and Xu, Junjun and Cheng, Yanzhe and Hu, Qinran},
  journal={Journal of Modern Power Systems and Clean Energy}, 
  title={Security Risk Assessment and Risk-oriented Defense Resource Allocation for Cyber-physical Distribution Networks Against Coordinated Cyber Attacks}, 
  year={2025},
  volume={13},
  number={1},
  pages={312-324},
  doi={10.35833/MPCE.2024.000288}}

@ARTICLE{Tang2021,
  author={Tang, Qingqing and Fei, Zesong and Li, Bin and Han, Zhu},
  journal={IEEE Internet of Things Journal}, 
  title={Computation Offloading in LEO Satellite Networks With Hybrid Cloud and Edge Computing}, 
  year={2021},
  volume={8},
  number={11},
  pages={9164-9176},
  doi={10.1109/JIOT.2021.3056569}}

@article{Xiu2025,
title = {Computation offloading and resource allocation in satellite edge computing networks: A multi-agent reinforcement learning approach},
journal = {Computer Networks},
volume = {272},
pages = {111680},
year = {2025},
issn = {1389-1286},
doi = {https://doi.org/10.1016/j.comnet.2025.111680},
url = {https://www.sciencedirect.com/science/article/pii/S1389128625006474},
author = {Qingxiao Xiu and Jun Liu and Xiangjun Liu and Jintao Wang}
}

@ARTICLE{Kato2019,
  author={Kato, Nei and Fadlullah, Zubair Md. and Tang, Fengxiao and Mao, Bomin and Tani, Shigenori and Okamura, Atsushi and Liu, Jiajia},
  journal={IEEE Wireless Communications}, 
  title={Optimizing Space-Air-Ground Integrated Networks by Artificial Intelligence}, 
  year={2019},
  volume={26},
  number={4},
  pages={140-147},
  doi={10.1109/MWC.2018.1800365}}

@article{Tedeschi2022,
title = {Satellite-based communications security: A survey of threats, solutions, and research challenges},
journal = {Computer Networks},
volume = {216},
pages = {109246},
year = {2022},
issn = {1389-1286},
doi = {https://doi.org/10.1016/j.comnet.2022.109246},
url = {https://www.sciencedirect.com/science/article/pii/S138912862200319X},
author = {Pietro Tedeschi and Savio Sciancalepore and Roberto {Di Pietro}}
}

@ARTICLE{Huang2025,
  author={Huang, Linan and Liu, Peilong and Chen, Xu and Jiang, Chunxiao and Kuang, Linling and Lu, Jianhua},
  journal={IEEE Internet of Things Journal}, 
  title={A Consolidated Game Framework for Cooperative Defense Against Cross-Domain Cyber Attacks in Satellite-Enabled Internet of Things}, 
  year={2025},
  volume={12},
  number={9},
  pages={12853-12868},
  doi={10.1109/JIOT.2024.3522558}}

\vfill

\end{document}